\documentclass[11pt,draftclsnofoot,onecolumn]{IEEEtran}
\usepackage[colorlinks,
linkcolor=red,
anchorcolor=green,
citecolor=green
]{hyperref} 
\usepackage{amsmath,amssymb}
\usepackage{latexsym}
\usepackage{CJK}
\usepackage{cite}

\usepackage{color, soul}

\usepackage{graphicx}
\usepackage{epstopdf}
\usepackage{epsfig}
\usepackage{subfigure} 

\usepackage{verbatim}  

\usepackage{mathrsfs}	
\usepackage{algorithm} 
\usepackage{algorithmic} 
\usepackage{booktabs}
\usepackage{textcomp}
\usepackage{multirow}
\usepackage{lettrine}

\usepackage{mathrsfs}
\usepackage{amsmath}
\usepackage{amssymb}

\usepackage{algorithm}
\usepackage{algorithmic}

\usepackage{amsthm} 

\newtheorem{thm}{Theorem}
\newtheorem{cor}{Corollary}
\newtheorem{lem}{Lemma}
\newtheorem{prop}{Proposition}
\newtheorem{defn}{Definition}
\newtheorem{rem}{Remark}

{ \bgroup
  \addtolength\abovedisplayshortskip{#1}
  \addtolength\abovedisplayskip{#1}
  \addtolength\belowdisplayshortskip{#1}
  \addtolength\belowdisplayskip{#1}}
{\egroup\ignorespacesafterend}

\usepackage{subeqnarray}
\usepackage{cases}
\usepackage{mathtools} 
\usepackage{bm}

\ifCLASSINFOpdf
\else
\fi
\begin{document}
%
\title{Amplifying Inter-message Distance: On Information Divergence Measures in Big Data}
%
%
%
\author{\IEEEauthorblockN{Rui She,
Shanyun Liu, and
Pingyi Fan, \IEEEmembership{Senior Member,~IEEE}
\\}

\thanks{
This work was supported by the China Major State Basic Research Development Program (973 Program)
No.2012CB316100(2), National Natural Science Foundation of China (NSFC) No. 61771283 and China Scholarship Council.

%
R. She, S. Liu, and P. Fan are with the Department of Electronic Engineering, Tsinghua University, Beijing, GA 30332 China (e-mail: sher15@mails.tsinghua.edu.cn; liushany16@mails.tsinghua.edu.cn; fpy@tsinghua.edu.cn).
%
%
}
}
\maketitle

\begin{abstract}
Message identification (M-I) divergence is an important measure of the information distance between probability distributions, similar to Kullback-Leibler (K-L) and Renyi divergence. In fact, M-I divergence with a variable parameter can make an effect on characterization of distinction between two distributions.
Furthermore, by choosing an appropriate parameter of M-I divergence, it is possible to amplify the information distance between adjacent distributions while maintaining enough gap between two nonadjacent ones. Therefore, M-I divergence can play a vital role in distinguishing distributions more clearly.
In this paper, we first define a parametric M-I divergence in the view of information theory and then present its major properties. In addition, we design a M-I divergence estimation algorithm by means of the ensemble estimator of the proposed weight kernel estimators, which can improve the convergence of mean squared error from ${O(\varGamma^{-j/d})}$ to ${O(\varGamma^{-1})}$ $({j\in (0,d]})$.
We also discuss the decision with M-I divergence for clustering or classification, and investigate its performance in a statistical sequence model of big data for the outlier detection problem.
\end{abstract}

\begin{IEEEkeywords}
Message Identification (M-I) Divergence, Discrete Distribution Estimation, Divergence Estimation, Big Data Analysis, Outlier Detection
\end{IEEEkeywords}

%
\IEEEpeerreviewmaketitle

\section{Introduction}
%
%
%
%
\IEEEPARstart{I}{n} the big data era, the amount of data from many kinds of areas is exploding greatly, and how to analyze the collected big data attracts more and more attention. For big data analysis,
there are a series of relevant technologies including machine learning, pattern recognition, statistics, estimation theory and so on.
As an essential element in machine learning, the information divergence can be used to deal with distribution problems by mapping the relationship between two probability distributions to nonnegative values. Currently, information divergences have been extended for nonnegative tensors and used to minimize the approximation error between the observed data and its model \cite{Learning-the-information-divergence}. Additionally, typical applications of information divergences also include faulty detection \cite{Analytical-model-of-the-KL-Divergence}, key frame selection \cite{Key-frame-selection-based-on-KL-divergence},
image and speech recognition \cite{Human-Ear-recognition-based-on-Multi-scale}, \cite{A-Study-on-Invariance-of- $f$-Divergence-and-Its-Application} and so on.

In the framework of the combination of information theory and big data analysis, information divergences were investigated as measures to handle the learning problem about distributions. In particular, the relative entropy as a special case of K-L divergence is a superior tool of measuring information distance in some applications such as anomaly detection \cite{Kullback-Leibler-Divergence-(KLD)-Based-Anomaly-Detection}, FMRI data processing \cite{Exploring-functional-connectivities-of-the-human-brain-using-multivariate-information- analysis}, clustering and classification \cite{Feature-selection-algorithm-for-hierarchical-text-classification}. 
Moreover, Shannon entropy can be also regarded as a special case of K-L divergence with an uniform distribution. It is appropriate for entropy to be applied to intrinsic dimension estimation \cite{On-local-intrinsic-dimension-estimation}, texture classification and image registration \cite{Applications-of-entropic}. In addition, information divergences can be applicable to extending methods of machine learning to distributional features
\cite{Divergence-estimation-for-multidimensional-densities}.

Although there are a great deal of available information divergences, little research is investigated on how to select a better one for a certain application to big data. Due to the different properties of different divergences, this issue is a significant work for information divergences used in big data analysis. Besides, another factor which can contribute to the issue is that a divergence-based estimation may depend on the selected divergence in a given task. Then we can see that it is flexible for information divergences to cope with different data learning tasks. For example, Euclidean distance has superior performance on handling data with Gaussian noise; K-L divergence is suitable for topic collection of text documents \cite{atent-Dirichlet-allocation}; Itakura-Saito divergence can perform well on audio signal processing \cite{Nonnegative-matrix-factorization}; as well as the MIM and non-parametric MIM which are similar to entropy, can be proven suitable for minority subset detection \cite{Message-Importance-Measure-and-Its-Application, Focusing-on-a-Probability-Element, Non-parametric-Message-Important-Measure}.
In addition, the information distance between given distributions can be also as a factor to make an effect on the divergence selection. Some divergences may not distinguish certain similar distributions due to the confusion between information distance and the statistical error.

In this paper, to study the information divergence as a measure in big data application, we will focus on the information distance measured by different divergences. As well, it is necessary to investigate the more efficient divergence estimation for practical applications or models. Before our work, let us review some typical information divergences first.

\subsection{Information Divergence measures}
There exist many different kinds of information divergences, which can play a vital role in the fields of information theory, statistics and big data analysis. To simply summarize a variety of information divergences, we focus on the most commonly used ones including K-L divergence, Renyi divergence and $\alpha$-, $\beta$- or $\gamma$-divergences
\cite{A-measure-of-asymptotic-efficiency,Differential-Geometrical-Methods,
Robust-paramater-estimation}, which belong to broader ones such as the f-divergences or Bregman divergences \cite{Divergence-measures-based-on-the-Shannon-entropy}.

For two finite discrete distributions $P=(p_1,p_2,..., p_n)$ and $Q=(q_1,q_2,..., q_n)$, the definitions of the popularly used divergences and some of their special cases or relationships are given below.

a). K-L divergence is defined as
\begin{equation}
  D(P\parallel Q)=\sum_i p_i\log(\frac{p_i}{q_i}).
\end{equation}

b). Renyi divergence of order $\alpha$ is defined as
\begin{equation}
  D_\alpha(P\parallel Q)=\frac{1}{\alpha-1}\log\sum_i p_i^{\alpha}q_i^{1-\alpha},
\end{equation}
where $0<\alpha< \infty$, and $\alpha \ne 1$. Notes that in the case of $\alpha\to 1$, Renyi divergence converges to K-L divergence.

c). $\alpha$-divergence is defined as
\begin{equation}
 \begin{aligned}
  D^{(\alpha)}(P\parallel Q)
  =\frac{\sum_i p_i^{\alpha} q_i^{1-\alpha}-\alpha p_i +(\alpha-1)q_i }{\alpha(\alpha-1)},
 \end{aligned}
\end{equation}
where $\alpha \to 1$, $\alpha=2$, $\alpha=-1$ and $\alpha=1/2$ denote K-L, Pearson Chi-square, inverse Pearson and double Hellinger distances, respectively.

d). $\beta$-divergence is defined as
\begin{equation}
 \begin{aligned}
  D^{(\beta)}(P\parallel Q)
  =\frac{\sum_i p_i^{\beta+1}+\beta q_i^{(\beta+1)} -(\beta+1)p_iq_i^{\beta} }{\beta(\beta+1)},
 \end{aligned}
\end{equation}
where $\beta=1$ and $\beta\to -1$ denote the Euclidean distance and Itakura-Saito divergence, respectively.

e). $\gamma$-divergence is defined as
\begin{equation}
 \begin{aligned}
   & D^{(\gamma)}(P\parallel Q)\\
   &= \frac{1}{\gamma(\gamma+1)} \Big[
   \log (\sum_i p_i^{\gamma+1}) + \gamma\log (\sum_i q_i^{\gamma+1})
   - (\gamma+1)\log (\sum_i p_iq_i^{\gamma}) \Big],
 \end{aligned}
\end{equation}
where K-L divergence becomes its special case if $\gamma \to 0$.

However, there may be also situations where the commonly used divergences can not work well. To this end, we introduce a new divergence different from the above divergences as follows.

\subsection{Message Identification Divergence}
In this subsection, we shall introduce a new parametric information identification measure, which is referred to as the message identification divergence (M-I divergence).
\begin{defn}\label{defn:M-I}
For two given probability distributions with a same finite alphabet, $P=(p_1, p_2, ... , p_n)$ and $Q=(q_1, q_2, ... , q_n)$, the M-I divergence with parameter $\varpi$ is defined as
\begin{equation}\label{M-I}
   D_{MI}(P\parallel Q,\varpi)=
   D_{\varpi}(P\parallel Q)=\log \sum_{i=1}^n{p_{i}e^{\left( \varpi \frac{p_i}{q_i} \right) } }-\varpi,
\end{equation}
where $\varpi > 0$ is an adjustable identification parameter.
\end{defn}

Note that the larger parameter $\varpi$ is, the larger contribution the information distance elements $p_i/q_i$ have to M-I divergence. In the application, it is necessary to set an appropriate $\varpi$ which is not too large to compute easily.

\subsection{Organization}
The rest of this paper is organized as follows. In Section II, we discuss some major properties of M-I divergence, such as its monotonicity, convexity and inequality. In Section III, we propose a multidimensional kernel estimator with the weight window, which can be adapted to estimate a discrete distribution. As well, we discuss its performance in the mean squared error (MSE) criterion. Then an ensemble estimator for M-I divergence is also proposed by use of some weighted-window kernel estimators. Section IV discuss how to use M-I divergence in big data analysis and apply it to a proposed outlier detection model. Besides, some simulations are also presented to check our theoretical results. Finally, we conclude the paper in Section V.

\section{The Properties of M-I Divergence}
In this section, some dominant properties of M-I divergence is investigated in details.\par
\subsection{The Non-negative Property}
\begin{prop}\label{prop:The Non-negative Property}
The M-I divergence $D_{\varpi}(P\parallel Q)$ with $\varpi > 0$ is non-negative for any probability $P$ and $Q$, namely
\begin{equation}
  D_{\varpi}(P\parallel Q)\ge 0.
\end{equation}
\end{prop}
\begin{proof}[{~~~Proof:}]
Define $f(x)=\exp{(\varpi x^{-1})}$ with $\varpi >0 $. It is readily seen that the second order derivative of $f(x)$ with respect to $x$ is positive, namely, $\frac{ \partial^{2} f(x) }{ \partial x^{2} }= \left( \varpi^2x^{-4}+2\varpi x^{-3} \right)e^{\varpi x^{-1}} >0$.
Then, we know that $f(x)$ is a convex function for $x \in \textbf{R}$.
According to Jensen's inequality and the concavity of function $\log (x)$, we have
\begin{equation}
  \begin{aligned}
    D_{\varpi}(P\parallel Q)
    & =  \log \sum_{i=1}^n{p_{i}e^{ \varpi \frac{p_i}{q_i} } }-\varpi
    \ge \log e^{ \varpi (\sum_{i=1}^n p_i ( \frac{q_i}{p_i} ) )^{-1} }-\varpi
    =  \log e^{\varpi}-\varpi = 0.
  \end{aligned}
\end{equation}
In particular, the equality holds if and only if $p_i=q_i$ $(i\in \{ 1, 2, ... , n \})$.
\end{proof}
\subsection{Monotonicity}
\begin{prop}\label{prop:The Increasing Property}
For the identification parameter $\varpi \in (0,+\infty)$, the M-I divergence $D_{\varpi}(P\parallel Q)$ is nondecreasing in $\varpi$.
\end{prop}
\begin{proof}[{~~~Proof:}]
By using the definition of M-I divergence and dividing its support set of $i$ into two parts, it is readily seen that the partial derivative of $D_{\varpi}(P|| Q)$ with respect to $\varpi$ satisfies
\begin{equation}
  \begin{aligned}
  \frac{ \partial D_{\varpi}(P\parallel Q)}{ \partial \varpi}
  & = \frac{ \sum\limits_{i \in \{ i: p_i \ge q_i\} \bigcup \{ i: p_i < q_i\} } {\left(  \frac{p_i}{q_i} -1\right)p_i e^{\left( \varpi \frac{p_i}{q_i} \right)} } }
   { \sum_i { p_i e^{\left( \varpi \frac{p_i}{q_i} \right)} } }\\
  & \ge \frac{ \sum\limits_{i \in \{ i: p_i \ge q_i\}} {\left(  \frac{p_i}{q_i} -1\right)p_i e^{\varpi } } }
   { \sum_i { p_i e^{\left( \varpi \frac{p_i}{q_i} \right)} } }
  + \frac{ \sum\limits_{i \in \{ i: p_i < q_i\}} {\left(  \frac{p_i}{q_i} -1\right)p_i e^{\varpi } } }
   { \sum_i { p_i e^{\left( \varpi \frac{p_i}{q_i} \right)} } }
    =
   \frac{ e^{ \varpi } \sum_i {\left(  \frac{p_i}{q_i} -1\right)p_i } }
   { \sum_i { p_i e^{\left( \varpi \frac{p_i}{q_i} \right)} } }.
  \end{aligned}
\end{equation}

According to Jensen's inequality, it is readily seen that $e^{ \varpi } \sum_i {(  \frac{p_i}{q_i} -1)p_i } \ge e^{\varpi}[ ( \sum_i p_i \frac{q_i}{p_i} )^{-1} -1].$
Thus, it can be readily verified that
\begin{equation}
 \begin{aligned}
   \frac{\partial D_{\varpi}(P\parallel Q)}{\partial \varpi}
   \ge \frac{ e^{\varpi}\left[ \left( \sum_i p_i \frac{q_i}{p_i} \right)^{-1} -1\right] }
   {\sum_i { p_i e^{\left( \varpi \frac{p_i}{q_i} \right)} } }
   = 0,
 \end{aligned}
\end{equation}
which means $D_{\varpi}(P\parallel Q)$ is monotonically nondecreasing in $\varpi$ and the property is proved.
\end{proof}
\begin{rem}
If and only if $p_i=q_i$ $(i\in \{ 1, 2, ... , n \})$, M-I divergence $D_{\varpi}(P \parallel Q)$ remains zero with increasing $\varpi$. In other cases $(p_i \ne q_i)$, $D_{\varpi}(P \parallel Q)$ is increasing in $\varpi$. According to this property, it can be apparently deduced that $\varpi$ is an adjustable parameter for M-I divergence to amplify the distance between different probability distributions.
\end{rem}
\subsection{The Convexity Property}
\begin{prop}\label{prop:The Convexity Property}
For any $\varpi>0$, M-I divergence $D_{\varpi}(P \parallel Q)$ is jointly convex in the case of exponential function. That is, for two given pairs of probability distributions $(P_0, Q_0)$ and $(P_1, Q_1)$ without zero elements, and any $0< \lambda <1$, we have
\begin{equation}
\left(1-\lambda \right)e^{D_{\varpi}(P_0 \parallel Q_0)}+\lambda e^{D_{\varpi}(P_1 \parallel Q_1)}
\ge e^{D_{\varpi}(P_{\lambda} \parallel Q_{\lambda})},
\end{equation}
where $P_{\lambda}=(1-\lambda)P_0+\lambda P_1$ and $Q_{\lambda}=(1-\lambda)Q_0+\lambda Q_1$.
\end{prop}
\begin{proof}[{~~~Proof:}]
Define $f(x)=xe^{\varpi x}$ with $\varpi>0$ and $x \in \textbf{R}$. It is easy to see that the first order and the second order derivative of $f(x)$ are both positive for $\varpi>0$ and $x>0$. Then, it is evident that $f(x)$ is convex for $\varpi>0$ and $x>0$. By using Jensen's inequality, in the case of $\varpi>0$, we have
\begin{equation}\label{equ:convex-1}
  \begin{aligned}
    \frac{ (1-\lambda)q_{0,i} }{q_{\lambda,i}}\left( \frac{p_{0,i}}{q_{0,i}} \right)e^{\varpi \left( \frac{p_{0,i}}{q_{0,i}} \right)}
    & +  \frac{ \lambda q_{1,i} }{q_{\lambda,i}}\left( \frac{p_{1,i}}{q_{1,i}} \right) e^{\varpi \left( \frac{p_{1,i}}{q_{1,i}} \right)}\\
    & \qquad \qquad \qquad \qquad \ge \left( \frac{p_{\lambda,i}}{q_{\lambda,i}} \right)e^{\varpi \left( \frac{p_{\lambda,i}}{q_{\lambda,i}} \right)},
  \end{aligned}
\end{equation}
where $p_{\lambda,i}=(1-\lambda)p_{0,i}+\lambda p_{1,i}$ and $q_{\lambda,i}=(1-\lambda)q_{0,i}+\lambda q_{1,i}$, as well as, $p_{m,i}$ and $q_{m,i}$ $(m=0, 1, \lambda ; ~ i=0, 1, ..., n)$ are any elements in $P_m$ and $Q_m$ $(m=0, 1, \lambda)$, respectively.

Then, for all elements of probability distributions $P_m$ and $Q_m$ $(m=0, 1, \lambda)$,  we have
\begin{equation}\label{equ:convex-2}
  \begin{aligned}
    (1-\lambda)\sum_i p_{0,i} e^{\varpi \left( \frac{p_{0,i}}{q_{0,i}} \right)}
    + & \lambda \sum_i p_{1,i} e^{\varpi \left( \frac{p_{1,i}}{q_{1,i}} \right)}
     \ge \sum_i p_{\lambda,i}e^{\varpi \left( \frac{p_{\lambda,i}}{q_{\lambda,i}} \right)},
  \end{aligned}
\end{equation}
for any $\varpi>0$, which proves the property.
\end{proof}
\begin{cor}
For any two pairs of probability distributions $(P_0, Q_0)$ and $(P_1, Q_1)$ without zero elements, and any $\lambda \in (0, 1)$, we have
\begin{equation}\label{equ:convex-3}
  \begin{aligned}
    \max\{ D_\varpi(P_0 \parallel Q_0),  D_\varpi(P_1 \parallel Q_1) \}
    \ge D_\varpi(P_\lambda \parallel Q_\lambda),
  \end{aligned}
\end{equation}
where $P_{\lambda}=(1-\lambda)P_0+\lambda P_1$ and $Q_{\lambda}=(1-\lambda)Q_0+\lambda Q_1$.
\end{cor}
\begin{proof}[{~~~Proof: }]
In view of the convexity property of M-I divergence $D_{\varpi}(P \parallel Q)$, in the case of exponential function, we have
\begin{equation}
  \begin{aligned}
    \left(1-\lambda \right)  \max \{ e^{D_{\varpi}(P_0 \parallel Q_0)}, e^{D_{\varpi}(P_1 \parallel Q_1)} \}
    +  \lambda \max \{ e^{D_{\varpi}(P_0 \parallel Q_0)}, e^{D_{\varpi}(P_1 \parallel Q_1)} \}
    \ge e^{D_{\varpi}(P_{\lambda} \parallel Q_{\lambda})}.
  \end{aligned}
\end{equation}
As a result, it can be easily seen that
\begin{equation}
  \begin{aligned}
     \max \{ e^{D_{\varpi}(P_0 \parallel Q_0)}, e^{D_{\varpi}(P_1 \parallel Q_1)} \}
    \ge e^{D_{\varpi}(P_{\lambda} \parallel Q_{\lambda})}.
  \end{aligned}
\end{equation}
Further, we can gain this corollary by use of the monotonicity of exponential function.
\end{proof}
\begin{cor}
For any probability distributions $P$, $Q_0$ and $Q_1$ which consist of positive elements, and $\lambda \in (0, 1)$, we have,
\begin{equation}
  \begin{aligned}
    \left(1-\lambda \right)e^{D_{\varpi}(P \parallel Q_0)}+\lambda e^{D_{\varpi}(P \parallel Q_1)}
    \ge e^{D_{\varpi}(P \parallel (1-\lambda)Q_0+\lambda Q_1)}.
  \end{aligned}
\end{equation}
\end{cor}
This can be verified by substituting $p$ for $p_0$ and $p_1$ in the convexity property.
\subsection{The Inequality Property}
\begin{prop}\label{prop:The Inequality Property}
For two given probability distributions with the finite support set, $P=(p_1, p_2, ... , p_n)$ and $Q=(q_1, q_2, ... , q_n)$ ($p_i>0,q_i>0, i=1, 2, ..., n $), the relationship among M-I divergence $D_\varpi(P \parallel Q)$, K-L divergence $D(P \parallel Q)$ and Renyi divergence $D_\alpha(P \parallel Q)$ can be indicated as
\begin{equation}\label{equ:Inquality_property}
  \begin{aligned}
    D_\varpi(P \parallel Q) \ge D(P \parallel Q) \ge D_\alpha(P \parallel Q),
  \end{aligned}
\end{equation}
where $\varpi \ge 1$ and $\alpha \in [0, 1]$.
\end{prop}
\begin{proof}[{~~~Proof: }]
Define a function $f(x)=e^{\varpi (x-1)}-x$ with $\varpi >0$ and $x \in (0, +\infty)$.
By setting $\frac{ \partial f(x)}{\partial x}=0$, it can be readily testified that the minimum of $f(x)$ is obtained at $x_0=1+\frac{1}{\varpi}\log{\frac{1}{\varpi}}$. Furthermore,
it is not difficult to see that only when $\varpi=1$, can $f(x_0, \varpi)=\frac{1}{\varpi}(1+\log\varpi)-1$ reach the maximum $f(x_0, \varpi=1)=0$. Therefore, it is clear to see that
\begin{equation}
    f(x)=e^{\varpi(x-1)}-x \ge 0,
\end{equation}
where $\varpi=1$, $x\in (0,+\infty)$ and $f(x)=0$ for $x=1$.

Now, the proof of left hand side inequality in Eq. (\ref{equ:Inquality_property}) can be cast into the proof of $ D_\varpi(P \parallel Q) \ge D(P \parallel Q)$ with $\varpi=1$.
This is due to the monotonicity of $D_\varpi(P \parallel Q)$ in Proposition \ref{prop:The Increasing Property}.

By averaging $f(\frac{p_i}{q_i})$ in the distribution $P=(p_1, p_2, ...., p_n)$ and considering the concavity of logarithmic function,
we have
$\log { \sum_i p_i e^{\varpi(\frac{p_i}{q_i}-1)} }-\log { \sum_i p_i \frac{p_i}{q_i}  } \ge 0$,
with $\varpi=1$.
As well, in virtue of Jensen's inequality,
it is apparent that
\begin{equation}\label{equ:inequality_1}
  \begin{aligned}
    \log { \sum_i p_i e^{\varpi(\frac{p_i}{q_i}-1)} }
    - \sum_i p_i \log { \frac{p_i}{q_i} } \ge 0,
  \end{aligned}
\end{equation}
which implies that $D_\varpi(P \parallel Q) \ge D(P \parallel Q)$ does work for $\varpi \ge 1$ due to the monotonicity of M-I divergence. 

Furthermore, $D_\alpha(P\parallel Q)$ keeps increasing in the order of $\alpha$, mentioned in Theorem 3 of \cite{Renyi-Divergence-and-Kullback-Leibler-Divergence}. Correspondingly, we have
\begin{equation}\label{equ:inequality_2}
  \begin{aligned}
    D(P \parallel Q)=\sup_{0< \alpha <1}D_\alpha(P\parallel Q).
  \end{aligned}
\end{equation}
Therefore, by combining Eq. (\ref{equ:inequality_1}) and Eq. (\ref{equ:inequality_2}), the inequality property of $D_\varpi(P \parallel Q)$ can be proved readily.
\end{proof}
\begin{rem}
According to the inequality property, the distance between two adjacent distributions can be amplified by the measure of M-I divergence. Moreover, M-I divergence is more sensitive than the other divergences to measure the distance between two nonadjacent distributions. Thus, it is more efficient for M-I divergence to distinguish two distributions.
\end{rem}
\section{Estimation of M-I Divergence}
\subsection{The Multidimensional Discrete Kernel Estimator}\label{III_A}
\subsubsection{Multidimensional Kernel with Weight Window}\ \par
With regard to the discrete kernel, there is a general definition to characterize it specifically
according to \cite{Discrete-triangular-distributions-and-nonparametric-estimation} as follows.
\begin{defn}\label{defn:2}
Let $\mathbb{S}$ be the finite support of the unknown probability mass function (p.m.f), to be estimated, with $x_i$ an element in $\mathbb{S}$. A p.m.f $K_{x_i,s}(\cdot)$ on support $\mathbb{S}_x$ (not depending on $s$) is regard as a discrete kernel with the parameter $s>0$, if it satisfies the following conditions:
\begin{subequations}
 \begin{align}
   &\begin{aligned}\label{equ:kernel_condition1}
      x_i \in \mathbb{S}_x,
    \end{aligned}\\
   &\begin{aligned}\label{equ:kernel_condition2}
      \lim_{s \to 0}\mathbb{E}(Z_{x_i,s})=x_i,
    \end{aligned}\\
   &\begin{aligned}\label{equ:kernel_condition3}
      \lim_{s \to 0}{\rm Var}(Z_{x_i,s})=0,
    \end{aligned}
 \end{align}
\end{subequations}
where $Z_{x_i,s}$ is a discrete random variable with p.m.f $K_{x_i,s}(\cdot)$.
\end{defn}
Based on the above characteristics of the discrete kernel, some special kernel functions can be designed in various ways. As well, we present a kernel estimator with the weight window for multidimensional discrete distribution as follow.
\begin{defn}\label{defn:3}
Let $\bm X_1 , \bm X_2, ..., \bm X_N$ be independent and identically distributed (i.i.d) multidimensional random variables with $d$-dimensional multivariate p.m.f $p(\bm x_i)$ on finite support $\mathbb{U}=[a_1, a_2, ..., a_L]^d$. A discrete kernel estimator $\widetilde p_s(\bm x_i)$ with a weight window is defined as
\begin{equation}\label{equ:weight_kernel}
  \begin{aligned}
    \widetilde p_s(\bm x_i) &= \frac{1}{N} \sum_{k=1}^N K_{\bm x_i,s,\tilde d}(\bm X_k)
                 = \frac{1}{N} \sum_{k=1}^N \sum_{\bm x_j \in \mathbb{U}} W(s, \bm x_i, \bm x_j) I_{\{\bm X_k = \bm x_j\}}, ~ \bm x_{i,j} \in \mathbb{U},
  \end{aligned}
\end{equation}
in which the weight window function is
\begin{equation}\label{equ:weight_W}
    W(s,\bm x_i, \bm x_j)= \left\{
     \begin{aligned}
        & 1-s, ~~~ ~~~~~ ~~~~ ~~~~ |x_i^u-x_j^u|=0\\
        & \frac{s}{(2\tilde{d}+1)^d-1}, 0<|x_i^u-x_j^u| \le \tilde d \\
        & 0, ~~~ ~~~~~ ~~~~ ~~~~ ~~~~~ |x_i^u-x_j^u| > \tilde d
  \end{aligned}
  \right.,
\end{equation}
where $u \in \{1, 2, ..., d\}$ denotes the dimension order, $s$ is a smoothing parameter, the window size $\tilde d$ is the distance of support indexes between $\bm x_i$ and $\bm x_j$ in every dimension, $L$ is the size of support indexes in every dimension, $I_{\{\bm X_k= \bm x_j\}}$ is the indicator function and $K_{ {\bm x_i},s,\tilde d}$ is the kernel function with $0 \le s <1$, $\bm x_i \in \mathbb{U}$ and $\tilde d<\frac{L-1}{2}$.
\end{defn}
From the above definition of multidimensional kernel estimator with the weight window, it can be clearly noted that when the smoothing parameter (or weight parameter) satisfies $s=0$, the discrete weight window function $W(s,\bm x_i, \bm x_j)$ degenerates into the indicator function $I_{\{\bm x_i= \bm x_j\}}$. As well, it is readily seen that regardless of weight parameter $s$ or variable $\bm x_i$, the sum of $W(s,\bm x_i, \bm x_j)$ for all $\bm x_j \in \{ \bm x_j: 0\le |x_i^u-x_j^u|\le \tilde d, \bm x_j \in \mathbb{U} \}$ fulfills
\begin{equation}
    \sum_{\bm x_j \in \mathbb{U}} W(s,\bm x_i, \bm x_j)=1.
\end{equation}

\begin{rem}\label{rem:tilde_p}
As far as the kernel estimator $\widetilde p_s(\bm x_i)$ is concerned, it is the core idea that relative frequencies derived from plug-in estimator are weighted to constitute the p.m.f estimator. In this way, more information of samples can be made use of to estimate every probability element in p.m.f. Furthermore, it is implied that the performance of the estimator $\widetilde p_s(\bm x_i)$ mainly depends on the selection of weight parameter in the case of the given p.m.f and known samples. In addition, if the weight parameter $s$ tends to zero as $N \to +\infty$, the estimator $\widetilde p_s(\bm x_i)$ will tend to the real $p(\bm x_i)$.
\end{rem}
\subsubsection{Selection of Kernel Weight Parameter}\ \par
We now consider selecting the weight parameter $s$, which can make an effect on the performance of kernel estimator. In general, the mean squared error (MSE) is accepted as a performance criterion for estimators. For a given $d$-dimensional multivariate p.m.f $p(\bm x_i)$ and its kernel estimator $\widetilde p_s(\bm x_i)$, the MSE can be treated as a function of $s$ as follows,
\begin{equation}\label{equ:MSE_1}
  \begin{aligned}
    & f_{\rm MSE}(s)\\
    & =\mathbb{E}\sum_{\bm x_i \in \mathbb{U} } \left[ \widetilde p_s(\bm x_i) -p(\bm x_i)\right]^2\\
    & =\sum_{\bm x_i \in \mathbb{U} }{\rm Var}\left[ \widetilde p_s(\bm x_i) \right]
      + \sum_{\bm x_i \in \mathbb{U} } {\rm Bias}^2\left[ \widetilde p_s(\bm x_i) \right]\\
    & = \frac{\sum_{\bm x_i \in \mathbb{U}}\sum_{\bm x_j \in \mathbb{U}} W^2(s, \bm x_i, \bm x_j)p(\bm x_j)}{N}
    - \frac{\sum_{\bm x_i \in \mathbb{U}} \left[ \sum_{\bm x_j \in \mathbb{U}} W(s, \bm x_i, \bm x_j)p(\bm x_j) \right]^2}{N}\\
    & + \sum_{\bm x_i \in \mathbb{U}} \left[ \sum_{\bm x_j \in \mathbb{U}} W(s, \bm x_i, \bm x_j)p(\bm x_j) -p(\bm x_i)\right]^2.
  \end{aligned}
\end{equation}

What is more,
it is readily realized that a value of $s$ can be provided by minimizing the MSE with respect to $s$. In that case, the optimal weight parameter is given by
\begin{equation}\label{equ:s_1}
  \begin{aligned}
   s_0=\arg \min_{0 \le s <1}f_{\rm MSE}(s).
  \end{aligned}
\end{equation}
By substituting Eq. (\ref{equ:weight_W}) into Eq. (\ref{equ:MSE_1}), we have
\begin{equation}\label{equ:MSE_2}
  \begin{aligned}
    f_{\rm MSE}(s)
    & = \frac{1-\sum_{\bm x_i \in \mathbb{U}}p^2(\bm x_i)}{N}+
     \frac{2}{N} \sum_{\bm x_i \in \mathbb{U}} \Bigg\{p^2(\bm x_i)-p(\bm x_i)
    -\frac{N p(\bm x_i)\sum_{\bm x_j \in \mathbb{V}_{\tilde d}}p(\bm x_j)}{(2\tilde d+1)^d -1} \Bigg\}s\\
    & + \frac{1}{N} \sum_{\bm x_i \in \mathbb{U}} \Bigg\{ (N-1)p^2(\bm x_i)
     + p(\bm x_i)+ \frac{\sum_{\bm x_j \in \mathbb{V}_{\tilde d}}p(\bm x_j)
    +(N-1)\left[\sum_{\bm x_j \in \mathbb{V}_{\tilde d}}p(\bm x_j) \right]^2 }{\left[ (2\tilde d+1)^d-1 \right]^2} \\
    &- \frac{2(N-1)p(\bm x_i)\sum_{\bm x_j \in \mathbb{V}_{\tilde d}}p(\bm x_j) }{ (2\tilde d+1)^d-1 } \Bigg\}s^2,
  \end{aligned}
\end{equation}
where the set $\mathbb{V}_{\tilde d}$ denotes $\{ \bm {\tilde x}_j: 0< |{\tilde x}_j^u-x_i^u|\le \tilde d, \bm {\tilde x}_j \in \mathbb{U} \}$.

By setting $\frac{\partial f_{\rm MSE}(s)}{\partial s}=0$, we can gain the minimum value of $f_{\rm MSE}(s)$. Therefore, it is readily seen that the optimal weight parameter is
\begin{equation}\label{equ:s_resolution}
  \begin{aligned}
    & s_0
    = \frac{1-\sum_{\bm x_i \in \mathbb{U}}p^2(\bm x_i)
    +\frac{\sum_{\bm x_i \in \mathbb{U}} \sum_{\bm x_j \in \mathbb{V}_{\tilde d}}p(\bm x_i)p(\bm x_j)}
    {(2\tilde d+1)^d -1}}{\varPhi(p,N,\tilde d)},
  \end{aligned}
\end{equation}
where the denominator is a function of $p$, $N$, and $\tilde d$, as
\begin{equation}\label{equ:s_denominator}
  \begin{aligned}
    & \varPhi(p,N,\tilde d)\\
    & =1+ \frac{\sum_{\bm x_i \in \mathbb{U}}\sum_{\bm x_j \in \mathbb{V}_{\tilde d}}p(\bm x_j)}
    {\left[(2\tilde d+1)^d -1\right]^2}\\
    & + (N-1)\Bigg\{ \sum_{\bm x_i \in \mathbb{U}}p^2(\bm x_i)
    + \frac{\sum_{\bm x_i \in \mathbb{U}}\left[\sum_{\bm x_j \in \mathbb{V}_{\tilde d}}p(\bm x_j)\right]^2}
    {\left[(2\tilde d+1)^d -1\right]^2}
    - \frac{2\sum_{\bm x_i \in \mathbb{U}}\sum_{\bm x_j \in \mathbb{V}_{\tilde d}}p(\bm x_i)p(\bm x_j)}
    {(2\tilde d+1)^d -1} \Bigg\}.
  \end{aligned}
\end{equation}\par
It is worth noting that the optimal weight parameter $s_0$ depends on the p.m.f $p(\bm x_i)$, the number of  multidimensional random variables $N$ and the distance $\tilde d$. However, the p.m.f is hardly known and needs estimating. For this reason, it can be thought over to replace the unknown $p(\bm x_i)$ with a consistent estimator, plug-in estimator $\hat p(\bm x_i)$. In that case, the suboptimal but practical solution $\hat s_0$ of weight parameter under the MSE criterion is put forward as
\begin{equation}\label{equ:hat_s_resolution}
  \begin{aligned}
    & \hat s_0
    = \frac{1-\sum_{\bm x_i \in \mathbb{U}}{\hat p}^2(\bm x_i)
    +\frac{\sum_{\bm x_i \in \mathbb{U}} \sum_{\bm x_j \in \mathbb{V}_{\tilde d}}{\hat p}(\bm x_i){\hat p}(\bm x_j)}
    {(2\tilde d+1)^d -1}}{\varPhi({\hat p},N,\tilde d)}.
  \end{aligned}
\end{equation}
\begin{rem}\label{rem:hat_s_resolution}
For a given p.m.f $p(\bm x_i)$ and $\tilde d$, it is readily seen that the suboptimal weight parameter satisfies $\hat s_0=O(1/N)$, which is similar to the optimal $s_0$. Moreover, if the weight parameter $s$ is replaced by $\hat s_0$, $f_{\rm MSE}(\hat s_0)$ will tend to zero as $N \to +\infty$. That is, the estimator $\widetilde p_{\hat s_0}(\bm x_i)$ tends to the real $p(\bm x_i)$ in the MSE criterion.
\end{rem}
\subsubsection{Performance Analysis}\ \par
In the view of MSE criterion, the multidimensional window kernel estimator $\widetilde p_{\hat s_0}(\bm x_i)$ keeps the same large-sample properties as plug-in estimator $\hat p(\bm x_i)$. However, it arises a question whether the former is superior to the latter in the performance of estimation. In order to distinguish which one is better, the measurements of MSE with respect to $\widetilde p_{\hat s_0}(\bm x_i)$ and $\hat p(\bm x_i)$ are given respectively by
\begin{equation}\label{equ:MSE_two_estimator}
  \begin{aligned}
    & {\rm MSE}\left( {\widetilde p_{\hat s_0}} \right)
    = f_{\rm MSE}(\hat s_0),\\
    & {\rm MSE}\left( {\hat {p}} \right)
    = \mathbb{E}\Big \{ \sum_{\bm x_i\in \mathbb{U}}\left[\hat p(\bm x_i)- p(\bm x_i)\right]^2 \Big \}.
  \end{aligned}
\end{equation}

Considering the definition of plug-in estimator, $\hat {p}(\bm x_i) = \frac{1}{N}\sum_{k=1}^N I_{\{\bm X_k = \bm x_i\}}$, we have
\begin{equation}\label{equ:MSE_hat_p}
  \begin{aligned}
    {\rm MSE}\left( {\hat {p}} \right)
    & = \sum_{\bm x_i\in \mathbb{U}}\left \{ \mathbb{E}\left[{\hat p}^2(\bm x_i)\right]- p^2(\bm x_i)\right \}\\
    & = \sum_{\bm x_i\in \mathbb{U}}\left[ \frac{p(\bm x_i)+(N-1)p^2(\bm x_i)}{N} - p^2(\bm x_i)\right]\\
    & = \frac{1-\sum_{\bm x_i\in \mathbb{U}}p^2(\bm x_i) }{N}.
  \end{aligned}
\end{equation}

By replacing $s$ with $\hat s_0$ in Eq. (\ref{equ:MSE_2}), the difference of the two MSE functions can be written as
\begin{equation}\label{equ:MSE_two_estimator2}
  \begin{aligned}
    {\rm MSE}\left( {\widetilde p_{\hat s_0}} \right)
    -{\rm MSE}\left( {\hat {p}} \right)
    & = \phi(\hat p,N, \tilde d){\hat s_0}
     + \psi(\hat p,N, \tilde d){\hat s_0}^2,
  \end{aligned}
\end{equation}
where $\phi(\hat p,N, \tilde d)$ and $\psi(\hat p,N, \tilde d)$ are both functions of $p$, $N$ and $\tilde d$. What is more, it is apparent that the convergence of $\phi(\hat p,N, \tilde d)$ and $\psi(\hat p,N, \tilde d)$ are $O(1)+O(1/N)$.

Due to the fact that the parameter $\hat s_0 \to 0$ as $N \to +\infty$, the $\psi(\hat p,N, \tilde d){\hat s_0}^2=O({\hat s_0}^2)$ tends to zero at a faster rate than $\phi(\hat p,N, \tilde d)\hat s_0$, as $N$ increasing. That is, the first term $\phi(\hat p,N, \tilde d)\hat s_0$ dominates the positive or negative nature of Eq. (\ref{equ:MSE_two_estimator2}).
In addition,
it is not difficult to know that $\phi(\hat p,N, \tilde d) < 0$ always holds by virtue of $p^2(\bm x_i)-p(\bm x_i)\le 0$ in the second term of Eq. (\ref{equ:MSE_2}).
Therefore, it is sure that for large enough $N$, ${\rm MSE}({\widetilde p_{\hat s_0}})<{\rm MSE}({\hat {p}})$ holds for any p.m.f ($p(\bm x_i) \ne 0$). This implies that $\widetilde p_{\hat s_0}(\bm x_i)$ has better performance than $\hat {p}(\bm x_i)$ in the MSE criterion.
\subsection{Multidimensional Weighted Ensemble Estimation}
For an ensemble of estimators $\{\hat D_{l_1}, \hat D_{l_2}, ..., \hat D_{l_{T}}\}$ of a parameter $D$, the weighted ensemble estimator with respect to the weight $\bm \lambda=\{\lambda_{l_1}, \lambda_{l_2}, ..., \lambda_{l_{T}}\}$ is defined as
\begin{equation}\label{equ:ensemble_estimator}
  \begin{aligned}
    \hat D_\lambda
    = \sum_{l\in \bar l} \lambda_{l}\hat D_l,
  \end{aligned}
\end{equation}
where the $\bar l=\{ l_1, l_2, ..., l_{T} \}$ denotes an index set. As well, the ensemble of weights is constrained by
$\sum_{l \in \bar l}\lambda_l=1$, which can ensure that the weighted ensemble estimator $\hat D_{\lambda}$ holds asymptotically unbiased in the case of the asymptotically unbiased estimators $\{\hat D_{l_1}, \hat D_{l_2}, ..., \hat D_{l_{T}}\}$.
\begin{thm}\label{thm_ensemble}
Assume the bias and variance of every estimator $ \hat D_l$ ($l\in \bar l$) satisfy the following conditions, respectively:
\begin{subequations}\label{equ:ensemble_C1}
 \begin{align}
    &\begin{aligned}
     {\rm Bias}(\hat D_l) = \sum_{j\in J}\gamma_j\varphi_j(l)\varGamma^{-j/2d}
    +\rho_{bias}(\varGamma^{-1/2} ),
    \end{aligned}\\
    &\begin{aligned}
    {\rm Var}(\hat D_l) = \rho_{var}(\varGamma^{-1/2} ),
    \end{aligned}
 \end{align}
\end{subequations}
where $\gamma_j$ are constants depending on a d-dimensional p.m.f $p(\bm x_i)$, $d$ is the dimension number, $J=\{j_i: 0<j_i\le d, 1\le i\le I ~~ and ~~ I\le T-1\}$ denotes an index set, $\varGamma$ is the number of samples, $\varphi_j(l)$ are independent functions of index $l$, and $\rho_{\tau}(\varGamma^{-1/2} )$ with any subscript $\tau$ are 
functions of $\varGamma^{-1/2}$.
Then, there exists a weight vector ${\bm \lambda^*}$ leading to
\begin{equation}\label{equ:ensemble_MSE}
    {\rm MSE}(\hat D_{\lambda^*}) = \mathbb{E}\left[ (\hat D_{\lambda^*}-D)^2 \right]
    \le \rho_{\lambda^*}(\varGamma^{-1/2} ).
\end{equation}
The weight vector ${\bm \lambda^*}$ is given by solving the following optimization problem:
\begin{equation}\label{equ:w_optimization}
\begin{aligned}
& {\bm \lambda^*} = \arg\min_{{\bm \lambda}}:{|| \bm \lambda ||}^2\\
&s.t.\left \{
\begin{aligned}
& \sum_{l\in \bar l}\lambda_l=1,\\
& \sum_{l\in \bar l}\lambda_l\varphi_{j}(l)=0, j\in J.
\end{aligned}\right.
\end{aligned}
\end{equation}
\end{thm}
\begin{proof}[{~~~Proof:}]
For the ensemble of estimators $\{ \hat D_l \}_{l\in \bar l}$, the bias of the weighted ensemble estimator $\hat D_{\lambda}$ is given by
\begin{equation}\label{equ:ensemble_bias}
\begin{aligned}
    {\rm Bias}(\hat D_{\lambda})
    & =\mathbb{E}\bigg \{ \sum_{l\in \bar l}\lambda_l \hat D_l -D \bigg \}
    = \sum_{l\in \bar l}\lambda_l {\rm Bias}(\hat D_l)\\
    & = \sum_{l\in \bar l}\sum_{j\in J}\gamma_j \lambda_l \varphi_j(l) \varGamma^{-j/2d}
    +\rho_{\lambda bias}(\varGamma^{-1/2} ).
\end{aligned}
\end{equation}

Considering the Cauchy-Schwartz inequality, it is not difficult to derive the variance of $\hat D_{\lambda}$ as follows:
\begin{equation}\label{equ:ensemble_var}
\begin{aligned}
    & {\rm Var}(\hat D_{\lambda})
    = \mathbb{E} \bigg \{ \sum_{l\in \bar l}\lambda_l [\hat D_l - \mathbb E(\hat D_l) ]\bigg \}^2
    \le \sum_{l\in \bar l}{\lambda_l}^2 \sum_{l\in \bar l}\mathbb{E}\left[ \hat D_l - \mathbb E(\hat D_l) \right]^2
    = ||\bm \lambda||^2{\rm Var}(\hat D_{l}).
\end{aligned}
\end{equation}\par
According to the condition $I \le T-1$, it is readily seen that there exists at least one solution for the constraint conditions of Eq. (\ref{equ:w_optimization}). Thus, there is a solution to minimize $||\bm \lambda||_2^2$, which can reduce the bias of the ensemble estimator to $\rho_{\lambda bias}(\varGamma^{-1/2} )$ and limit the contribution of the variance. Then, the MSE of ensemble estimator with respect to the optimal solution $\bm \lambda^*$ can be derived as
\begin{equation}\label{equ:ensemble_MSE 2}
\begin{aligned}
    {\rm MSE}(\hat D_{\lambda^*})
    & = {\rm Bias}^2(\hat D_{\lambda^*})+{\rm Var}(\hat D_{\lambda^*})\\
    & \le \rho_{\lambda^* bias}^2(\varGamma^{-1/2} )+ ||\bm \lambda^*||^2\rho_{\lambda^* var}(\varGamma^{-1/2} )\\
    & = \rho_{\lambda^*}(\varGamma^{-1/2} ),
\end{aligned}
\end{equation}
which can verify the theorem.
\end{proof}
In addition, from Theorem \ref{thm_ensemble}, it is not difficult to derive the corollary \ref{cor_ensemble} by replacing functions $\rho_{\tau}(\cdot)$ with order $O(\cdot)$ or $o(\cdot)$ as follows.
\begin{cor}\label{cor_ensemble}
For the bias and variance of the ensemble of estimators $\{ \hat D_l\}_{l\in \bar l}$, the following conditions are satisfied as
\begin{subequations}\label{equ:ensemble_cor_C1}
 \begin{align}
    &\begin{aligned}\label{equ:ensemble_cor_C1_a}
    {\rm Bias}(\hat D_l) =  \sum_{j\in J}\gamma_j\varphi_j(l)\varGamma^{-j/2d}
    + o\left(\varGamma^{-1/2}\right),
    \end{aligned}\\
    &\begin{aligned}\label{equ:ensemble_cor_C1_b}
    {\rm Var}(\hat D_l) =O\left(\varGamma^{-1}\right).
    \end{aligned}
 \end{align}
\end{subequations}
Then, there exists a weight vector $\bm \lambda^*$ given by Eq. {\rm (\ref{equ:w_optimization})}, which can lead to
\begin{equation}\label{equ:ensemble_cor_MSE}
    {\rm MSE}(\hat D_{\lambda^*})= O\left(\varGamma^{-1}\right).
\end{equation}
In order to obtain the above convergence rate of MSE, it is sufficient for
$\sum_{l\in \bar l}\lambda_l\varphi_j(l)\varGamma^{-j/2d} $ to be of order $O(\varGamma^{-1/2})$. Thus, the optimal weight vector can be determind as
\begin{equation}\label{equ:w_optimization_cor}
\begin{aligned}
& {\bm \lambda^*} = \arg\min_{\bm \lambda} \delta\\
&s.t.\left \{
\begin{aligned}
& \sum_{l\in \bar l}\lambda_l=1,\\
& \big|\sum_{l\in \bar l}\lambda_l\varphi_{j}(l)\varGamma^{1/2-j/2d}\big| \le \delta,\quad j\in J,\\
& ||\bm \lambda||_2^2\le \epsilon,
\end{aligned}\right.
\end{aligned}
\end{equation}
where the parameter $\epsilon$ is small enough.
\end{cor}
\begin{rem}
For the weighted ensemble estimator, on the one hand, it possesses a distinctive superiority that the MSE is endowed with faster convergence by using the weight vector to eliminate the higher order bias terms. On the other hand, it is visible that the weighted estimator applies to the circumstance where there are estimators with different indexes.
\end{rem}
\subsection{Ensemble Estimator for M-I Divergence}
In this subsection, we focus on the estimation of M-I divergence between two $d$-dimensional multivariate distributions $P$ and $Q$ whose p.m.fs are $p(\bm x_i)$ and $q(\bm x_i)$ with the known finite support $\mathbb{U}=[a_1, a_2, ..., a_L]^d$. In terms of the definition of M-I divergence, it is apparent that the estimator of M-I divergence depends on the estimator of $F_{\varpi}(P \parallel Q)=\sum_{\bm x_i\in \mathbb{U}} p(\bm x_i)e^{ \varpi \frac{p(\bm x_i)}{q(\bm x_i)} } $, which can be approximately calculated by using the samples splitting approach as follows.

Assume that the i.i.d. random samples from $P$ are divided into two parts $\{ \bm X_1, ..., \bm X_N\}$ and $\{ \bm X_{N+1}, ...,$ $\bm X_{N+M}\}$. The latter part is used to estimate the p.m.f of $P$ at the $N$ points $\{ \bm X_1, ..., \bm X_N \}$ by means of the weighted-window kernel. Similarly, the weighted-window kernel estimator of the p.m.f of $Q$ at the $N$ points $\{ \bm X_1, ..., \bm X_N \}$ is calculated by use of the i.i.d. samples $\{ \bm Y_1, ..., \bm Y_M\}$ drawn from $Q$. Then the estimator of $F_{\varpi}(P\parallel Q)$ can be written as
\begin{equation}\label{equ:hat_F_d}
  \hat F_{\tilde d}
  = \frac{1}{N} \sum_{i=1}^N e^{ \varpi \frac{\widetilde p_{\hat s_0}(\bm X_i)}
  {\widetilde q_{\hat s_0}(\bm X_i)} },
\end{equation}
where $\widetilde p_{\hat s_0}$ and $\widetilde q_{\hat s_0}$ are weighted-window kernel estimators with the distance $\tilde d$ mentioned in the subsection \ref{III_A}.

\begin{thm}\label{thm_bias}
The bias of the estimator $\hat F_{\tilde d}$ with weighted-window kernel is given by
\begin{equation}\label{equ:theorem_Bias}
\begin{aligned}
    & {\rm Bias}(\hat F_{\tilde d})
    = \sum_{j\in J} b_j\Big(\frac{K}{M}\Big)^{j/d}
    + o\big(\frac{K}{M}\big)+ o\big(\frac{1}{K}\big)+O\big(\frac{1}{M}\big),
\end{aligned}
\end{equation}
where $K= [\frac{(2\tilde d+1)}{L}]^d M$ is a real number determined by $\tilde d <\frac{L-1}{2}$ and
the parameter $b_j$ are constants depending on the distributions $P$ and $Q$.
\end{thm}
\begin{thm}\label{thm_var}
The variance of the estimator $\hat F_{\tilde d}$ with weighted-window kernel is given by
\begin{equation}\label{equ:theorem_Var}
\begin{aligned}
    & {\rm Var}(\hat F_{\tilde d})
    = O\Big(\frac{1}{N}\Big).
\end{aligned}
\end{equation}
\end{thm}
The proof of Theorem \ref{thm_bias} and \ref{thm_var} are given in Appendix \ref{Appendix_thm_bias} and \ref{Appendix_thm_var}.

For a positive $T \ge I + 1 $ and a positive real number set $\bar l=\{ l_1, ..., l_T\}$, let $K= \varDelta\sqrt{M}$, $\varDelta=[\frac{(2\tilde d_l+1)}{L}]^d \sqrt{M}$, $l\in \bar l$, $M=\mu\varGamma$ and $N=(1-\mu)\varGamma$ with $0< \mu <1$. Note that the $l$ indexes over the distance size $\tilde d$ for the weighted-window kernel estimator.
Then, the ensemble estimator of $\hat F_{\tilde d}$ is given by
\begin{equation}\label{equ:ensemble_F}
\begin{aligned}
    \hat F_\lambda=\sum_{l\in \bar l}\lambda_l \hat F_{\tilde d_l},
\end{aligned}
\end{equation}
where $\hat F_{\tilde d_l}$ denotes a $\hat F_{\tilde d}$ with an index $l$ for different $\tilde d$.

From Theorem \ref{thm_bias} and \ref{thm_var}, it is readily seen that the biases and variances of the ensemble estimators satisfy the conditions mentioned in Eq. (\ref{equ:ensemble_cor_C1_a}) when $\varphi_j(l)=\varDelta^{j/d}$ and $J=\{ 1, ..., d\}$. Therefore, it is available to find the optimal $\bm \lambda^*$ by using Corollary \ref{cor_ensemble} so that we can improve the MSE convergence for the estimation of $F_{\varpi}(P||Q)$. That is, we can make good use of the better estimator $\hat F_{\lambda^*}$ to obtain the better estimator of M-I divergence as
\begin{equation}\label{equ:ensemble_M-I}
\begin{aligned}
    &\hat D_\varpi(P\parallel Q)=\log \hat F_{\lambda^*} -\varpi.
\end{aligned}
\end{equation}
In addition, it is easily to see that the {\rm MSE} of $\hat D_\varpi(P\parallel Q)$ is given by
\begin{equation}\label{equ:ensemble_MSE_hat_D}
\begin{aligned}
    {\rm MSE}(\hat D_\varpi(P\parallel Q))=O(\varGamma^{-1}),
\end{aligned}
\end{equation}
whose proof is given in Appendix \ref{Appendix_MSE_M-I}.

In order to summarize the above process more specifically, the ensemble estimator with weighted-window kernel for M-I divergence is listed in Algorithm \ref{alg_ensemble}.
\begin{algorithm}
\caption{Optimally weighted ensemble estimator with window kernel for M-I divergence}\label{alg_ensemble}
\begin{algorithmic}[1]
\REQUIRE $\mu$, $\epsilon$, $L$, the distance set $\{\tilde d_l:l\in \bar l \}$ with the index set $\bar l$, samples set $\{\bm Y_1, ..., \bm Y_M\}$ from distribution $Q$, samples set $\{\bm X_1, ..., \bm X_{\varGamma}\}$ from distribution $P$, dimension $d$, the parameter $\varpi$ of M-I divergence.
\ENSURE The optimally weighted ensemble estimator $\hat D_\varpi$.
\STATE $M \leftarrow \mu\varGamma$, $N \leftarrow \varGamma - M$
\STATE calculate the $\bm \lambda^*$ by use of Eq. (\ref{equ:w_optimization_cor}) with $\varphi_j(l)=\varDelta^{j/d}$, $\varDelta=[\frac{(2\tilde d_l+1)}{L}]^d \sqrt{M}$, $j\in \{1, ..., d\}$
and $l\in \bar l$.
\FORALL {$l \in \bar l$}
    \STATE $\tilde d_l$ $\leftarrow$ choosing a distance by $l \in \bar l$
    \FORALL {$\bm X_i \in \{\bm X_1, ..., \bm X_N\}$}
        \STATE $\hat p(\bm X_i)$ and $\hat q(\bm X_i)$ $\leftarrow$ $\frac{1}{M}\sum_{k=N+1}^{N+M} I_{\{\bm X_k=\bm X_i\}}$ and
        $\frac{1}{M}\sum_{k=1}^M I_{\{\bm Y_k=\bm X_i\}}$, respectively.
        \STATE smoothing parameters $\hat s_0(p)$ and $\hat s_0(q)$ $\leftarrow$ using Eq. (\ref{equ:hat_s_resolution}) with $M$, $\tilde d_l$, $\hat p(\bm X_i)$ and
        $\hat q(\bm X_i)$.
        \STATE $\widetilde p_{\hat s_0(p)}(\bm X_i)$ and $\widetilde q_{\hat s_0(q)}(\bm X_i)$ $\leftarrow$ make full use of $\hat s_0(p)$, $\hat s_0(q)$, $\hat p(\bm X_i)$ and
        $\hat q(\bm X_i)$ to
        calculate the two weighted-window kernel estimators by adopting Eq. (\ref{equ:weight_kernel}) in definition \ref{defn:3}.
    \ENDFOR
    \STATE $\hat F_{\lambda_l}\leftarrow \frac{1}{N} \sum_{i=1}^N e^{ \varpi \frac{\widetilde p_{\hat s_0(p)}(\bm X_i)}{\widetilde q_{\hat s_0(q)}(\bm X_i)} }$
\ENDFOR
\STATE $\hat F_{\lambda^*}\leftarrow \sum_{l\in \bar l}\lambda_l^* \hat F_{\tilde d_l}$
\STATE $\hat D_\varpi\leftarrow \log\hat F_{\lambda^*}-\varpi$
\end{algorithmic}
\end{algorithm}
\section{Application to Big Data Analysis}
In this section, we will discuss how to exploit divergence measures to classify or cluster the data belonging to different distributions.
In particular, we take into account the following detection problem about outlier or minority sequences in a set of sequences.
\subsection{The Model with Unknown Number of Outliers}
Assume that among a number of sample sequences, there are an unknown number of outlier sequences to be detected. The i.i.d. samples in the typical sequences are gained from a known distribution $P_{t}$, while in the outlier sequences, the i.i.d. samples are taken from an unknown distribution $ Q_{f}$. In order to design a test to detect the outlier sequences, it is necessary to construct a model applying to the problem.

Consider $T_0$ independent sequences ($T_0\ge 3$), each of which can be denoted by $ \mathcal{ X}^{(i)}$ for $i=1, ..., T_0$. As well, each $\mathcal{X}^{(i)}$ consists of $\varGamma_0$ i.i.d samples $\{ {\bm X}^{(i)}_1, ..., {\bm X}^{(i)}_{\varGamma_0} \}$ drawn from either a typical distribution $P_{t}$ or an unknown outlier distribution $Q_{f}$. It notes that there may exist $k_0$ numbers of outlier sequences, where the integer $k_0\in [0, \frac{T_0}{2})$ is uncertain. As well, the notation ${\bm X}^{(i)}_{k}$ denotes the $k$-th sample in the $i$-th sequence. Furthermore, by comparing the empirical typical distribution $\hat P_t$ with the distribution estimation $\hat P(\mathcal{X}^{(i)})$ for every $\mathcal{X}^{(i)}$, we have the following test as
\begin{equation}
 \mathcal{F} (\hat P(\mathcal{X}^{(i)}); \hat P_{t})\to
 \left\{
  \begin{aligned}
    \mathcal{F} (\hat P_{t_i}; \hat P_t),  & \quad \mathcal{X}^{(i)}\in \mathcal{M}_t  \\
    \mathcal{F} (\hat Q_{f_i}; \hat P_t),  & \quad \mathcal{X}^{(i)}\in \mathcal{M}_f
  \end{aligned}\right.,
\end{equation}
where $\mathcal{F}(f_1;f_2)$ denotes a measurement between two distribution $f_1$ and $f_2$, $\hat P_{t_i}$ and $\hat Q_{f_i}$ are estimations with respect to $P_t$ and $Q_f$, $\mathcal{M}_t$ and $\mathcal{M}_f$ denote the typical sequences set and the outlier sequences set, respectively.

In practice, our sequence model for outliers detection is applicable for the case in which the outlier distributions $Q_{f}$ is unknown a priori, whereas the typical distribution $P_{t}$ or at least the empirical distribution $\hat P_t$ is known. This is rational for many practical scenarios, in which systems regularly start without any outliers and it is easy to possess sufficient information for $P_{t}$. In addition, the study of such a model can apply to many applications, such as vacant channel detection in cognitive wireless networks, fraud and anomaly detection in large data sets, state monitoring in sensor networks and so on.

\subsection{Outlier Detection with Divergence Measures}
Considering the performance of divergence measures on distinguishing different distributions, we can use the information distances measured by divergences to detect the outliers in the above sequence model. The method of outliers detection based on the sequence model is designed as following.

We make use of the i.i.d. samples to estimate the M-I divergence between a pending sequence and the typical sequence. The M-I divergence estimation can be applicable to the outlier sequence model as a measurement for clustering. Furthermore, a clustering algorithm such as k-means can be adopted to distinguish the outlier sequences from the typical ones. The above process is more specifically summarized in Algorithm \ref{alg_outlier detection}. Similarly, it is feasible to design the outliers detection methods with other divergence measures such as K-L divergence and Renyi divergence.

\begin{figure}[!t]
\centering
\subfigure[ROC curve of different divergences]{\includegraphics[width=3.0in]{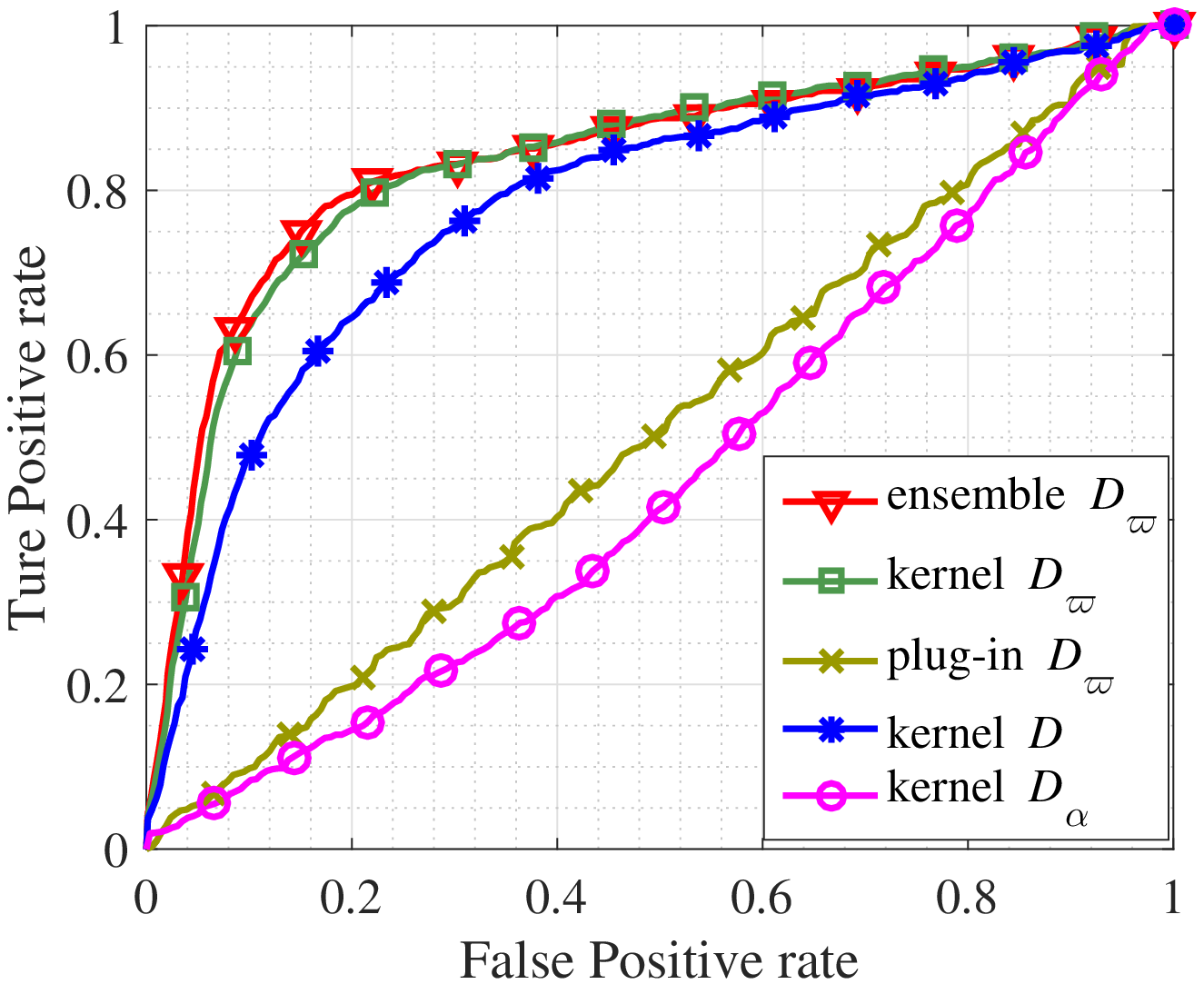}%
\label{fig_ROC of different divergences}}
\hfil
\subfigure[AUC of different divergences]{\includegraphics[width=3.0in]{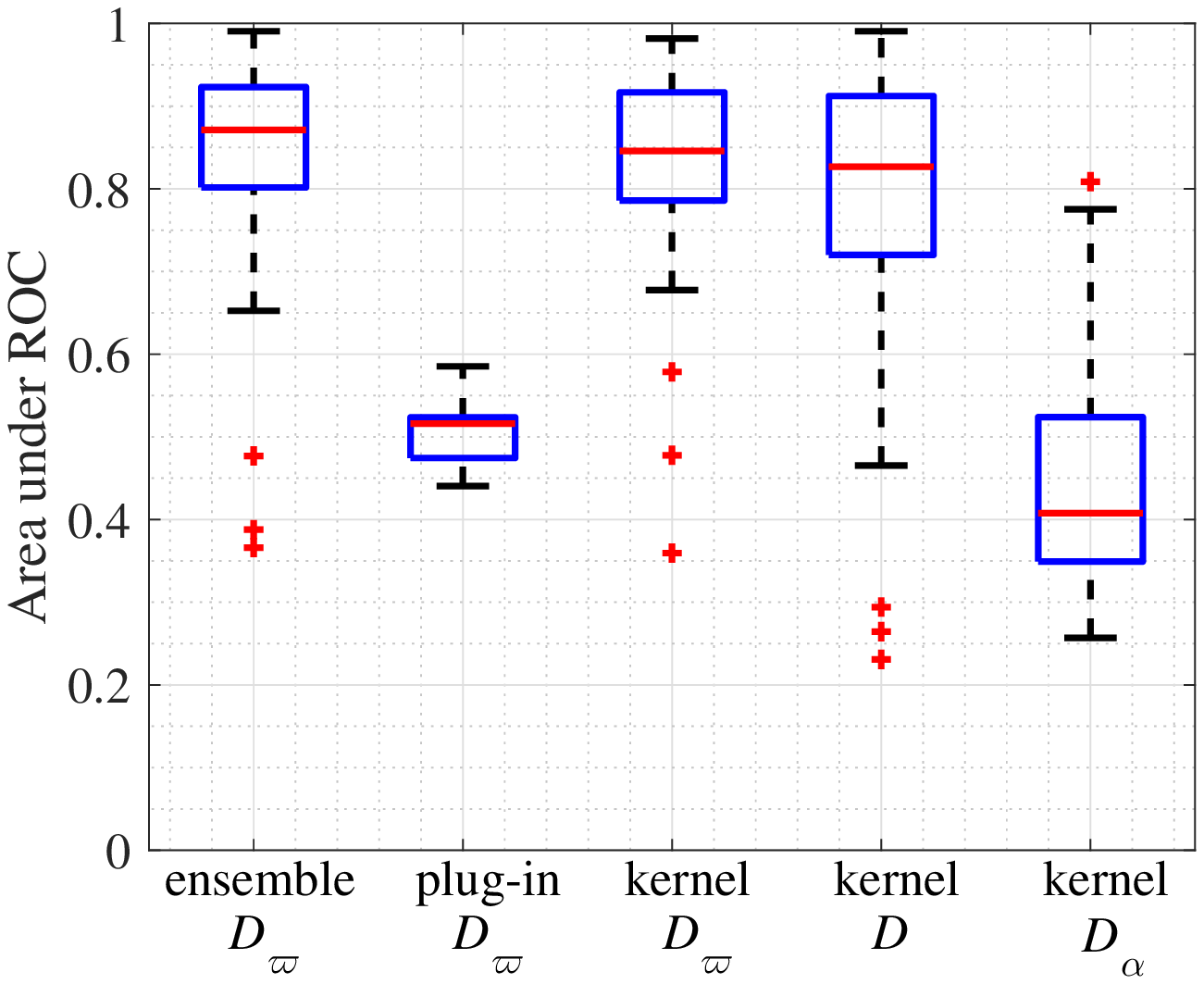}%
\label{fig_AUC_box}}
\hfil
\subfigure[F-score of different divergences]{\includegraphics[width=3.0in]{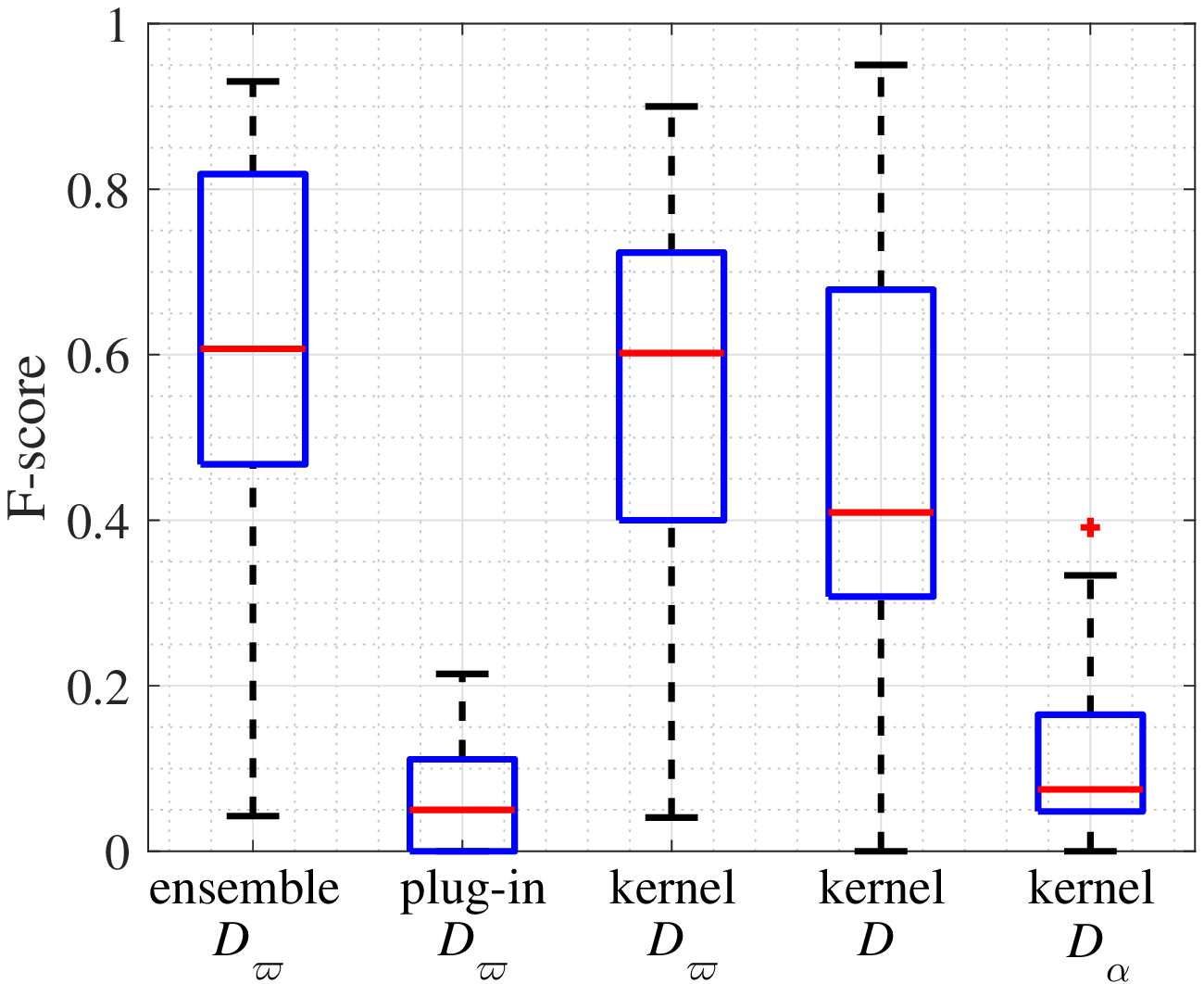}%
\label{fig_Fscore_box}}
\caption{Performance of different divergences in the example with each sequence size $\varGamma_0=6000$, sequences number $T_0=200$, outlier sequences number $k_0=20$ and the number of experiments $N_{T_0}=100$. }
\label{fig_performance}
\end{figure}
\begin{figure}[!t]
\centering
\includegraphics[width=3.1in]{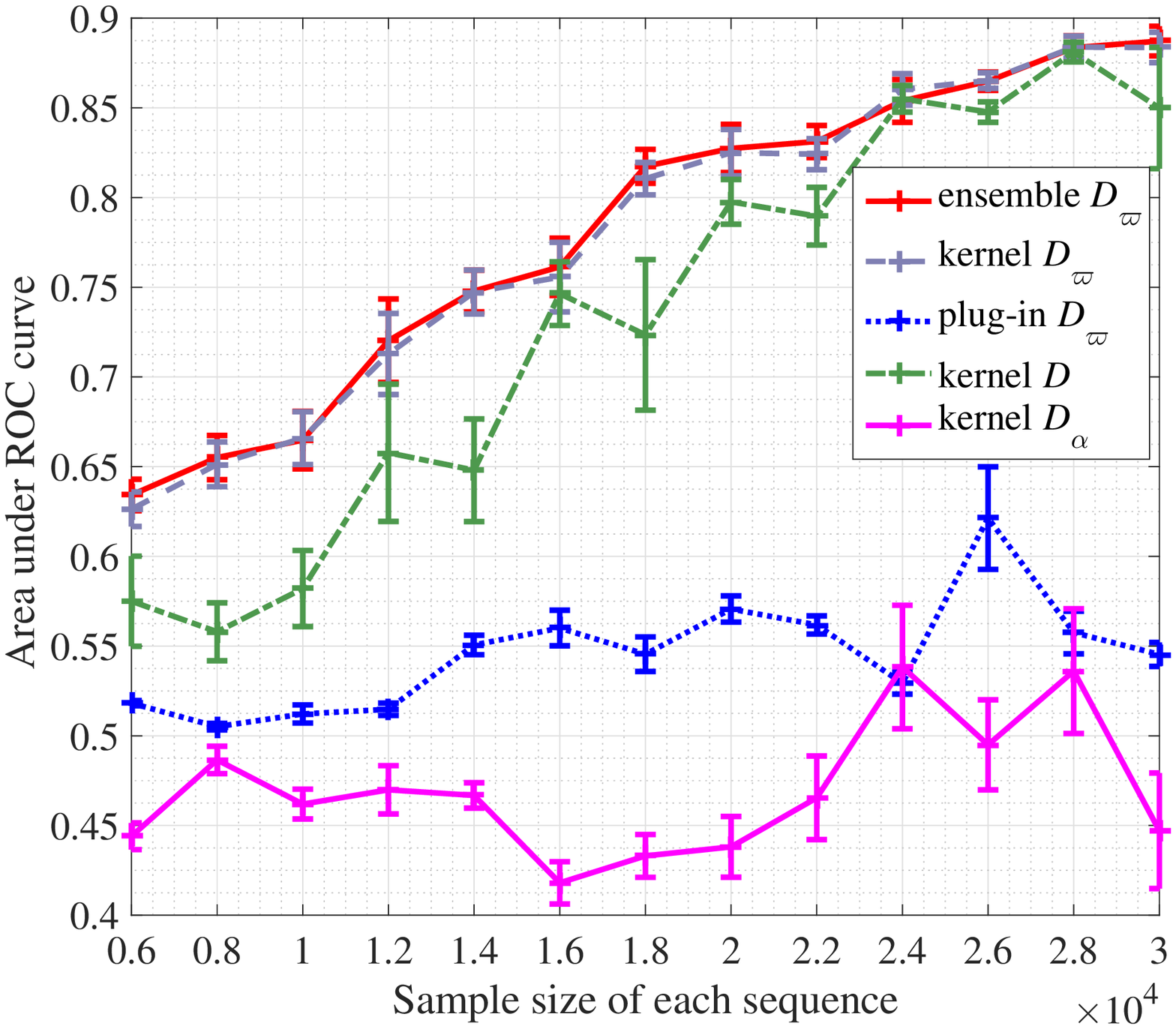}
\caption{ Means and variances of AUC with respect to different sample size in the example with sequences number $T_0=200$, outlier sequences number $k_0=20$ and the number of experiments $N_{T_0}=100$.}
\label{fig_AUC_T}
\end{figure}
To demonstrate the divergence measures' availability on outlier detection, we utilize the sequence model with unknown number of outliers to characterize a kind of outlier detection scenario. As an example, we regard a binomial distribution $\mathcal{B}(n_0, p_a)$ as the typical distribution, in which the probability is denoted as
$p_{t_i}=C^{i}_{n_0} p_a^i (1-p_a)^{n_0-i}$ with $n_0=11$ and $p_a=0.45$. By contrast, the i.i.d. samples drawn from another $\mathcal{B}(n_0, p_b)$ with $p_b=0.445$ are regarded as outliers.
Then, we can randomly generate a set of sample sequences, among which each sequence consists of $\varGamma_0$ i.i.d samples drawn from either $\mathcal{B}(n_0, p_a)$ or $\mathcal{B}(n_0, p_b)$. In terms of the sequence model, our goal is to detect the $k_0$ outlier sequences in the set of $T_0$ sequences.
\begin{algorithm}
\caption{Outlier Sequences Detection with M-I divergence}\label{alg_outlier detection}
\begin{algorithmic}[1]
\REQUIRE Pending sequences $\{ {\bm X}^{(i)}_1, ..., {\bm X}^{(i)}_{\varGamma_0} \}$ for $i=1, ..., T_0$, the typical sequence $\{ {\bm Y}_1, ..., {\bm Y}_{\varGamma_0} \}$ obeying a known empirical distribution $\hat P_t$, the parameter $\mu$, $\varpi$
\ENSURE The results of the outlier detection, $\mathcal{X}^{(i)} \in \mathcal{M}_f$
\STATE Divide $\varGamma_0$ into two parts: $M_0 \leftarrow \mu \varGamma_0$ ($0 <\mu<1 $), $N_0 \leftarrow \varGamma_0 - M_0$
\FORALL {$ i \in [1, T_0]$}
\STATE $\hat D_{\varpi,i}(\mathcal{X}^{(i)}||\mathcal{Y})$, $\hat D_{\varpi,i}(\mathcal{Y}||\mathcal{X}^{(i)})$  $\leftarrow$
calculate the estimator of M-I divergence between $\{ {\bm Y}_{1}, ..., {\bm Y}_{M_0} \}$ and $\{ {\bm X}^{(i)}_{1}, ..., {\bm X}^{(i)}_{\varGamma_0} \}$ as well as the estimator of M-I divergence between $\{ {\bm X}^{(i)}_{1}, ..., {\bm X}^{(i)}_{M_0} \}$ and $\{ {\bm Y}_{1}, ..., {\bm Y}_{\varGamma_0} \}$ respectively, according to Algorithm \ref{alg_ensemble}.
\STATE $\hat D_{\varpi,i}$ $\leftarrow$ $\{\hat D_{\varpi,i}(\mathcal{X}^{(i)}||\mathcal{Y})+\hat D_{\varpi,i}(\mathcal{Y}||\mathcal{X}^{(i)})\}/2$
\ENDFOR
\STATE Divide $\hat D_{\varpi,i}$ into the normal set $\mathcal{\hat M}_t$ or the outlier set $\mathcal{\hat M}_f$ $\leftarrow$ select a clustering algorithm such as k-means.
\STATE $\mathcal{X}^{(i)} \in \mathcal{M}_f$ $\leftarrow$ $\hat D_{\varpi,i} \in \mathcal{\hat M}_f$
\end{algorithmic}
\end{algorithm}

In order to illustrate the performance of M-I divergence on outlier detection,
we deal with the above example by means of Algorithm \ref{alg_outlier detection}. Besides, we also replace the ensemble estimator and M-I divergence in that algorithm with other estimators and divergences to make comparisons, as shown in Fig. \ref{fig_performance} and Fig. \ref{fig_AUC_T}. In our simulation, we choose the parameter $\mu=1/2$, $\varpi=1$, $\epsilon=3d$ and the distance set $\{\tilde d_l =l ; l=1,2 \}$ for the weighted ensemble estimator of M-I divergence. As well, the weight window is set to $\tilde d = 1$ in kernel estimator, which is used to estimate the discrete distribution in K-L divergence, Renyi divergence (with $\alpha =1/2$) and M-I divergence (with $\varpi =1$).

From Fig. \ref{fig_performance}, it is seen that M-I divergence outperforms K-L divergence and Renyi divergence by using the same estimator, which matches the inequality property well.
Besides, this experiment shows that M-I divergence performs better by using weighted ensemble estimator than other estimators, owing to the convergence improvement. Moreover, due to the smaller reduction in the convergence term with large samples, the kernel estimator is close to the ensemble estimator for M-I divergence.

The Fig. \ref{fig_AUC_T} shows that the result of outliers detection tends to be more precise as sample size increasing for M-I and K-L divergence estimated by the ensemble or kernel estimator. However, that is not dramatically improved for Renyi and M-I divergence with plug-in estimator, which may result from the small distinction between the typical distribution $\mathcal{B}(n_0, p_a)$ and the outlier distribution $\mathcal{B}(n_0, p_b)$.
In short, it still illustrates that M-I divergence can distinguish two closing distributions more clearly rather than K-L divergence and Renyi divergence.

\section{Conclusion}
In this paper, we investigated the information distance problem and proposed a parametric information divergence, i.e., M-I divergence, which measures the distinction between two discrete distributions, similar to K-L divergence and Renyi divergence. Furthermore, M-I divergence has its own dramatic properties on amplifying the distance between adjacent distributions while maintaining enough gap between two nonadjacent ones. This makes M-I divergence as a promising decision making tool for the statistical big data analysis. We have investigated several fundamental properties of M-I divergence, and proposed a multidimensional kernel estimator with a weight window to estimate probability distributions in M-I divergence. Furthermore, we also presented a M-I divergence estimation algorithm by means of the weighted ensemble estimator with the window kernel. In addition, we have investigated the performance of M-I divergence on decision making of classification or clustering and applied it to design an algorithm about the outlier detection problem. In the future, we plan to investigate a parameter selection method for M-I divergence and design algorithms for other practical applications in big data.

%
%

%
%
%
\appendices
\section{Proof of Theorem \ref{thm_bias}}\label{Appendix_thm_bias}
Define $f(x)=e^{\varpi x^{-1}}$ with $\varpi > 0$ and $x > 0$. Note that
${\rm Bias}(\hat F_{\tilde d})= \mathbb{E}[f(\frac{\widetilde q_{\hat s_0}(\bm X)}{\widetilde p_{\hat s_0}(\bm X)})-f(\frac{\bar q(\bm X)}{\bar p(\bm X)})]
+\mathbb{E}[f(\frac{\bar q(\bm X)}{\bar p(\bm X)})- f(\frac{q(\bm X)}{p(\bm X)})]$.
In order to find bounds for these terms, the Taylor series expansion of $f(\frac{\widetilde q_{\hat s_0}(\bm X)}{\widetilde p_{\hat s_0}(\bm X)})$ around $\frac{\bar q(\bm X)}{\bar p(\bm X)}$ and $f(\frac{\bar q(\bm X)}{\bar p(\bm X)})$ around $\frac{q(\bm X)}{p(\bm X)}$ are given by, respectively,

\begin{subequations}
 \begin{align}
 &\begin{aligned}\label{equ:bias1_expension}
    f\left(\frac{\widetilde q_{\hat s_0}(\bm X)}{\widetilde p_{\hat s_0}(\bm X)}\right)
    & =\sum_{i=0}^2 \frac{f^{(i)}\left( \frac{\bar q(\bm X)}{\bar p(\bm X)} \right)}{i!} \tilde h^{i}(\bm X)\\
    & \qquad \qquad \qquad + \frac{1}{6}f^{(3)}(\tilde \xi_{\bm X}) \tilde h^{3}(\bm X),
  \end{aligned}\\
  &\begin{aligned}\label{equ:bias2_expension}
    f\left(\frac{\bar q(\bm X)}{\bar p(\bm X)}\right)
    & =\sum_{i=0}^2 \frac{f^{(i)}\left( \frac{q(\bm X)}{p(\bm X)} \right)}{i!} \check h^{i}(\bm X)\\
    & \qquad \qquad \qquad + \frac{1}{6}f^{(3)}(\check \xi_{\bm X}) \check h^{3}(\bm X),
  \end{aligned}
 \end{align}
\end{subequations}
where $\tilde \xi_{\bm X}\in \left( \frac{\bar q(\bm X)}{\bar p(\bm X)},
\frac{\widetilde q_{\hat s_0}(\bm X)}{\widetilde p_{\hat s_0}(\bm X)} \right)$ and
$\check \xi_{\bm X}\in \left( \frac{q(\bm X)}{p(\bm X)},
\frac{\bar q(\bm X)}{\bar p(\bm X)} \right)$
come from the mean value theorem, $\tilde h(\bm X)$ denotes
$ \frac{\widetilde q_{\hat s_0}(\bm X)}{\widetilde p_{\hat s_0}(\bm X)}
-\frac{\bar q(\bm X)}{\bar p(\bm X)}$ and $\check h(\bm X)$ denotes
$ \frac{\bar q_{\hat s_0}(\bm X)}{\bar p_{\hat s_0}(\bm X)}
-\frac{q(\bm X)}{p(\bm X)}$. As well, the $\bar q(\bm x_i)$ is defined as

\begin{equation}\label{equ:estimator2_pq}
  \begin{aligned}
    & \bar q(\bm x_i)=\frac{\hat U_{\tilde d, q}({\bm x}_i)}{V_{\tilde d}}
    =\frac{\sum_{j=1}^M I_{\{\bm Y_j\in \mathbb{U}_{\tilde d, \bm x_i} \} }}{M V_{\tilde d}},
  \end{aligned}
\end{equation}
where $\bm x_i \in \mathbb{U}$, and $V_{\tilde d}=(2\tilde d+1)^d$ is the number of set $\mathbb{U}_{\tilde d, \bm x_i}=\{ \bm x_j: |x_j^u-x_i^u|\le \tilde d, \bm x_j \in \mathbb{U} \}$. Similarly, the $\bar p(\bm x_i)$ can be calculated in the same way.

\begin{lem}\label{lem:E_e1^m_p_q}
Let a $d$-dimension variable $\bm X$ be a realization of p.m.f $p$ independent of the window kernel estimators $\widetilde q_{\hat s_0}$ and $\widetilde p_{\hat s_0}$.
As well, p.m.f $p$ and $q$ are on the same support $\mathbb{U}=[a_1, ..., a_L]^d$.
Then, for a subscript $z$ denoting $p$ or $q$, we have
\begin{subequations}\label{equ:E_e1^m_p_q}
  \begin{align}
    &
    \begin{aligned}
    & \tilde e_z(\bm X)\\
    & = \sum_{i=1}^d \tilde {\hat c}_{e_z,i,\tilde d}(\bm X)\Big(\frac{K}{M}\Big)^{i/d}+ o\Big(\frac{K}{M}\Big)
    + O\Big(\frac{1}{M}\Big),
    \end{aligned}\\
    &
    \begin{aligned}
    & \mathbb{E}[\tilde e_z^m(\bm X)]\\
    & = \left(\sum_{i=1}^d \hat c_{e_z^m,i,\tilde d}\Big(\frac{K}{M}\Big)^{i/d}\right)^m
    +o\Big(\frac{K}{M}\Big)
    + O\Big(\frac{1}{M}\Big),
    \end{aligned}
  \end{align}
\end{subequations}
where $K=\varDelta \sqrt{M}$ with $\varDelta=[\frac{(2\tilde d_l+1)}{L}]^d \sqrt{M}$, $\tilde e_{p}(\bm X)$ and $\tilde e_{q}(\bm X)$ denote \{$\tilde p_{\hat s_0}(\bm X)-\bar p(\bm X)$\} and \{$\tilde q_{\hat s_0}(\bm X)-\bar q(\bm X)$\} respectively, as well as, $\hat c_{e_z^m,i,\tilde d}(\bm X)$ is a function of $\bm X$ and $\tilde d$, $\hat c_{e_z^m,i,\tilde d}$ is a function of $d$.
\end{lem}
\begin{proof}[{~~~Proof:}]
As far as $\tilde e_p(\bm X)$ is concerned, it is easy to see
\begin{equation}\label{equ:e1_p}
  \begin{aligned}
    \tilde e_p(\bm X)
    & =\widetilde p_{\hat s_0}(\bm X)-\bar p(\bm X)\\
    & = \frac{1}{N} \sum_{k=1}^N \left \{ \sum_{\bm x_j \in \mathbb{U}} W(\hat s_0, \bm X, \bm x_j) I_{\{\bm X_k = \bm x_j\}}-\frac{I_{\{\bm X_k\in \mathbb{U}_{\tilde d, \bm X} \} }}{V_{\tilde d}} \right \}\\
    & = \frac{1}{N} \sum_{k=1}^N \left \{ I_{\{\bm X_k = \bm X\}}-\frac{I_{\{\bm X_k\in \mathbb{U}_{\tilde d, \bm X} \} }}{V_{\tilde d}} \right \}
     -\hat s_0\left \{ \frac{1}{N}\sum_{k=1}^N \Bigg \{ \frac{\sum_{\substack{\bm x_j\in \mathbb{U}_{\tilde d, \bm X}\\ \bm x_j \ne X }} I_{\{\bm X_k = \bm x_j \}} }{V_{\tilde d}-1}
    - I_{\{\bm X_k = \bm X \}}\Bigg \} \right \}\\
    & =\hat p(\bm X)-\frac{\hat U_{\tilde d,p}(\bm X)}{V_{\tilde d}} + O(\hat s_0),
  \end{aligned}
\end{equation}
where $\hat U_{\tilde d,p}(\bm X)= \frac{1}{M} \sum_{k=1}^M I_{\{\bm X_k\in \mathbb{U}_{\tilde d, \bm X} \} }$.

Assume that the continuous probability density function $\hat f_p(\bm x)$ has continuous partial derivatives of order $d$. By use of Taylor series expansion, we can easily have the integral with respect to $\hat f_p$ in the integral range $\mathbb{B}_{\tilde d, \bm x}$ as
\begin{equation}\label{equ:expansion_PU}
  \begin{aligned}
    \hat U_{\tilde d,\hat f_p}(\bm x)
    &= \int_{\mathbb{B}_{\tilde d,\bm x}}\hat f_p(\bm z)d{\bm z}\\
    & =   \hat f_p(\bm x) V_{\tilde d,\hat f_p}(\bm x)+\sum_{i=1}^d \hat c_{i,\tilde d}(\bm x)V_{\tilde d,\hat f_p}^{1+i/d}(\bm x)
    + o\big( V_{\tilde d,\hat f_p}^2(\bm x) \big),
  \end{aligned}
\end{equation}
where $V_{\tilde d,\hat f_p}(\bm x)=\int_{\mathbb{B}_{\tilde d, \bm x}}dz$ is the volume of set $\mathbb{B}_{\tilde d, \bm x}$ and $\hat c_{i,\tilde d}$ depends on $\tilde d$ and $\hat f_p$.

Let the continuous set $\mathbb{B}_{\tilde d,\bm x}=\{ \bm x: x^u \in (x_{i-\tilde d -1}^u,  x_{i+\tilde d}^u)  \}$ ($u\in \{1,..., d\}$ denotes $u$-th dimension) correspond to
the discrete set $\mathbb{U}_{\tilde d,\bm x_i}=\{\bm x_j: | x_j^u - x_i^u|\le \tilde d, \bm x_j\in \mathbb{U} \}$, which means $V_{\tilde d,\hat f_p}(\bm x)=\int_{\mathbb{B}_{\tilde d, \bm x}}dz=V_{\tilde d}$.
Then, $\hat f_p(\bm x)$ fulfills the following conditions,
\begin{equation}\label{equ:expansion_hat_fp_condition}
\left \{
  \begin{aligned}
  & \int_{\bm x_{i-1}}^{\bm x_i}\hat f_p(\bm z)d\bm z=\hat p(\bm x_i)\\
  & \hat f_p(\bm x_i)=\hat p(\bm x_i),
  \end{aligned}
\right.
\end{equation}
which implies that
$\hat U_{\tilde d,\hat f_p}(\bm x_i)= \int_{\mathbb{B}_{\tilde d,\bm x_i}}\hat f_p(\bm z)d{\bm z}
=\hat U_{\tilde d,p}(\bm x_i)$.

According to the Eq. (\ref{equ:expansion_PU}), it is easy to see that the $\hat U_{\tilde d,p}(\bm X)$ can be expanded as ,
\begin{equation}\label{equ:expansion_hat_U}
  \begin{aligned}
    \hat U_{\tilde d,p}(\bm X)
    & = \int_{\mathbb{B}_{\tilde d,\bm X}}\hat f_p(\bm z)d{\bm z}\\
    & =   \hat p(\bm X) V_{\tilde d}+\sum_{i=1}^d \hat c_{i,\tilde d}(\bm X)V_{\tilde d}^{1+i/d}
    + o\left( V_{\tilde d}^2 \right),
  \end{aligned}
\end{equation}
where $V_{\tilde d}=(2\tilde d+1)^d=\frac{K}{M}L^d$.
Considering Remark \ref{rem:hat_s_resolution}, namely $\hat s_0=O(\frac{1}{M})$, and Eq. (\ref{equ:e1_p}), it can be readily seen that
\begin{equation}\label{equ:e1_p_solution}
  \begin{aligned}
    \tilde e_p(\bm X)
    = \sum_{i=1}^d \tilde {\hat c}_{e_p,i,\tilde d}(\bm X)\Big(\frac{K}{M}\Big)^{i/d}+ o\Big(\frac{K}{M}\Big)
    + O\Big(\frac{1}{M}\Big).
  \end{aligned}
\end{equation}

Furthermore, according to the binomial theorem, we have
\begin{equation}\label{equ:Ee1^m_p_solution}
  \begin{aligned}
    & \mathbb{E}[\tilde e_p^m(\bm X)]\\
    & =\mathbb{E}\bigg \{ \Big [
    \sum_{i=1}^d \tilde {\hat c}_{e_p,i,\tilde d}(\bm X)\Big(\frac{K}{M}\Big)^{i/d}
     + o\Big(\frac{K}{M}\Big)
    + O\Big(\frac{1}{M}\Big)
    \Big ]^m
    \bigg \}\\
    & = \left(\sum_{i=1}^d \hat c_{e_p^m,i,\tilde d}\Big(\frac{K}{M}\Big)^{i/d}\right)^m +o\Big(\frac{K}{M}\Big)
    + O\Big(\frac{1}{M}\Big).
  \end{aligned}
\end{equation}
Similarly, $\tilde e_q(\bm X)$ and $\mathbb{E}[\tilde e_q^n(\bm X)]$ can also be derived.
Therefore, it is apparent that Lemma \ref{lem:E_e1^m_p_q} has been testified.
\end{proof}
\begin{lem}\label{lem:E_h^t}
For a $d$-dimension variable $\bm X$ which denotes a realization of p.m.f $p$ independent of the window kernel estimators $\widetilde p_{\hat s_0}$ and $\widetilde q_{\hat s_0}$, we have
\begin{subequations}\label{equ:expension_h}
  \begin{align}
    & \tilde h(\bm X)
    =\sum_{i=1}^d \hat c_{h,i,\tilde d}(\bm X)\Big(\frac{K}{M}\Big)^{i/d}
    + o\big(\frac{K}{M}\big)+O\big(\frac{1}{M}\big), \\
    & \mathbb{E}[\tilde h^{t}(\bm X)]
    = \sum_{i=1}^d c_{\tilde h^t,i,\tilde d}\Big(\frac{K}{M}\Big)^{i/d}
    + o\big(\frac{K}{M}\big)+O\big(\frac{1}{M}\big),
  \end{align}
\end{subequations}
where 
$\tilde h(\bm X)$ denotes
$ \frac{\widetilde q_{\hat s_0}(\bm X)}{\widetilde p_{\hat s_0}(\bm X)}
-\frac{\bar q(\bm X)}{\bar p(\bm X)}$, $c_{\tilde h^t,i,\tilde d}$ is a function of $\tilde d$, $p$ and $q$, as well as, $\hat c_{h,i,\tilde d}(\bm X)$ is a function of $\bm X$, $\tilde d$, $p$ and $q$.
\end{lem}
\begin{proof}[{~~~Proof:}]
By expanding the $\frac{\widetilde q_{s_0}(\bm X)}{\widetilde p_{s_0}(\bm X)}$around
$\bar q(\bm X)$ and $\bar p(\bm X)$, we can readily have
\begin{equation}\label{equ:E_hat_h}
  \begin{aligned}
    & \tilde h(\bm X)\\
    & =\frac{\widetilde q_{\hat s_0}(\bm X)}{\widetilde p_{\hat s_0}(\bm X)}
    - \frac{\bar q(\bm X)}{\bar p(\bm X)}\\
    & =\frac{\tilde e_q(\bm X)}{\bar p(\bm X)}
    - \frac{\bar q(\bm X)}{\bar p^2(\bm X)}\tilde e_p(\bm X)
     - \frac{\tilde e_q(\bm X)\tilde e_p(\bm X)}{\bar p^2(\bm X)}
    + \frac{\bar q(\bm X)}{\bar p^3(\bm X)}\tilde e_p^2(\bm X)\\
    & + \frac{\bar q(\bm X)}{\bar p^4(\bm X)}\tilde e_q(\bm X)\tilde e_p^2(\bm X)
     + o(\tilde e_p^2(\bm X)+ \tilde e_q(\bm X)\tilde e_p^2(\bm X)),
  \end{aligned}
\end{equation}
where $\tilde e_{p}(\bm X)$ and $\tilde e_{q}(\bm X)$ denote \{$\tilde p_{\hat s_0}(\bm X)-\bar p(\bm X)$\} and \{$\tilde q_{\hat s_0}(\bm X)-\bar q(\bm X)$\}, respectively.

According to Lemma \ref{lem:E_e1^m_p_q}, it is not difficult to see that
\begin{equation}
  \begin{aligned}
    & \tilde h(\bm X)\\
    &= \sum_{i=1}^d \hat c_{h_1,i,\tilde d}(\bm X)\Big(\frac{K}{M}\Big)^{i/d}
     + \bigg(\sum_{i=1}^{d} \hat c_{h_2,i,\tilde d}(\bm X)\Big(\frac{K}{M}\Big)^{i/d}\bigg)^2
     + \bigg(\sum_{i=1}^{d} \hat c_{h_3,i,\tilde d}(\bm X)\Big(\frac{K}{M}\Big)^{i/d}\bigg)^3 \\
    & + o \left( \bigg(\sum_{i=1}^{d} \hat c_{h_3,i,\tilde d}(\bm X)
    \Big(\frac{K}{M}\Big)^{i/d}\bigg)^3  \right)
    + o\Big(\frac{K}{M}\Big)+O\Big(\frac{1}{M}\Big)\\
    & = \sum_{i=1}^d \hat c_{h,i,\tilde d}(\bm X)\Big(\frac{K}{M}\Big)^{i/d}
    + o\Big(\frac{K}{M}\Big)+O\Big(\frac{1}{M}\Big).
  \end{aligned}
\end{equation}

Furthermore, by applying the binomial theorem, we have
\begin{equation}\label{equ:E_hat_h^t}
  \begin{aligned}
    & \mathbb{E}[\tilde h^t(\bm X)]\\
    & = \left(\sum_{i=1}^d \hat {\tilde c}_{\tilde h^t,i,\tilde d}
    \Big(\frac{K}{M}\Big)^{i/d}\right)^{t}
    + o\left(\frac{K}{M}\right)+O\left(\frac{1}{M}\right)\\
    & = \sum_{i=1}^d I_{\{i\ge t\}}\hat c_{\tilde h^t,i,\tilde d}\Big(\frac{K}{M}\Big)^{i/d}
     + o\left(\frac{K}{M}\right)+O\left(\frac{1}{M}\right),
  \end{aligned}
\end{equation}
which can indicate that the lemma is proved.
\end{proof}

According to Lemma \ref{lem:E_h^t} and Eq. (\ref{equ:bias1_expension}), it is readily seen that
\begin{equation}\label{equ:bias1_expension_solution}
  \begin{aligned}
    & \mathbb{E}\left[f\left(\frac{\widetilde q_{\hat s_0}(\bm X)}{\widetilde p_{\hat s_0}(\bm X)}\right)
    - f\left(\frac{\bar q(\bm X)}{\bar p(\bm X)}\right)\right]\\
    & =\mathbb{E}\left[ \sum_{i=1}^2\tilde c_{\tilde h_i,\tilde d}(\bm X)\tilde h^i(\bm X)
     + o\left(\tilde c_{\tilde h_3,\tilde d}(\bm X)\tilde h^2(\bm X)\right)\right]\\
    & = \sum_{i=1}^d \tilde c_{\tilde h_1,i,\tilde d}\Big(\frac{K}{M}\Big)^{i/d}
    + \bigg(\sum_{i=1}^d \tilde c_{\tilde h_2,i,\tilde d}\Big(\frac{K}{M}\Big)^{i/d}\bigg)^2
    + o\bigg( \bigg(\sum_{i=1}^d \tilde c_{\tilde h_3,i,\tilde d}\Big(\frac{K}{M}\Big)^{i/d} \bigg)^2 \bigg)
     + o\Big(\frac{K}{M}\Big)+ O\Big(\frac{1}{M}\Big)\\
    & = \sum_{i=1}^d c_{\tilde h,i,\tilde d}\Big(\frac{K}{M}\Big)^{i/d}
    + o\Big(\frac{K}{M}\Big)+ O\Big(\frac{1}{M}\Big).
  \end{aligned}
\end{equation}

\begin{lem}\label{lem:E_bar_e^m_p_q}
Let a $d$-dimension variable $\bm X$ be a realization of p.m.f $p$ independent of $\bar p$ and $\bar q$ mentioned in Eq. (\ref{equ:estimator2_pq}).
Then it can be given that
\begin{subequations}\label{equ:E_e2^m_p_q}
  \begin{align}
    &\begin{aligned}
     \mathbb{E}[\bar e_p^m(\bm X)]
    & =\mathbb{E}\left[ (\bar p(\bm X)-p(\bm X))^m \right]\\
    & = \sum_{i=1}^d \bar c_{e_p^m,i,\tilde d}\Big(\frac{K}{M}\Big)^{i/d}
    +o\Big(\frac{K}{M}\Big)+o\Big( \frac{1}{ \sqrt{M} } \Big)+  o\Big( \frac{1}{K} \Big),
    \end{aligned}\\
    &\begin{aligned}
    \mathbb{E}[\bar e_q^n(\bm X)]
    & =\mathbb{E}\left[ (\bar q(\bm X)-q(\bm X)^n\right]\\
    & = \sum_{i=1}^d \bar c_{e_q^n,i,\tilde d}\Big(\frac{K}{M}\Big)^{i/d}
    +o\Big(\frac{K}{M}\Big)+o\Big( \frac{1}{ \sqrt{M} } \Big)+ o\Big( \frac{1}{K} \Big),
    \end{aligned}
  \end{align}
\end{subequations}
where $\bar c_{e_p^m,i,\tilde d}$ and $\bar c_{e_q^n,i,\tilde d}$ are functions of 
$\tilde d$.
\end{lem}
\begin{proof}[{~~~Proof:}]
Define $U_{\tilde d, p}(\bm X)=\mathbb{E}[\frac{1}{M}\sum_{i=1}^M I_{\{\bm X_i\in \mathbb{U}_{\tilde d,\bm X}\}}]=\sum_{\bm x_j\in \mathbb{U}_{\tilde d,\bm X}} p(\bm x_j)$.
On the one hand, for $\hat e_p(\bm X)=\bar p(\bm X)-\frac{U_{\tilde d,p}(\bm X)}{V_{\tilde d}}$, it is readily seen that
\begin{equation}\label{equ:E_hat_e_p_1}
  \begin{aligned}
  & \mathbb{E}[\hat e_p(\bm X)]=\mathbb{E}\left[\frac{\hat U_{\tilde d, p}(\bm X)}{V_{\tilde d}}-\frac{U_{\tilde d,p}(\bm X)}{V_{\tilde d}}\right]=0.\\
  \end{aligned}
\end{equation}

What is more, 
we have
\begin{equation}\label{equ:E_hat_e_p_2}
  \begin{aligned}
    & {\rm Var}\Big(\frac{\hat U_{\tilde d, p}(\bm X)}{V_{\tilde d}}\Big)\\
    & = \frac{\mathbb{E}[\sum_{i=1}^M I_{ \{X_i \in \mathbb{U}_{\tilde d, \bm X}\} }]}{M^2V_{\tilde d}^2}
    + \frac{\sum_{i=1}^M\sum_{\substack {j=1,\\ j \ne i}}^{M}
    I_{ \{X_i \in \mathbb{U}_{\tilde d, \bm X}\} }I_{ \{X_j \in \mathbb{U}_{\tilde d, \bm X}\} }}
    {M^2V_{\tilde d}^2} - \frac{U_{\tilde d,p}^2(\bm X)}{V_{\tilde d}^2}\\
    & =\frac{U_{\tilde d,p}(\bm X)}{MV_{\tilde d}^2} - \frac{U_{\tilde d,p}^2(\bm X)}{M^2V_{\tilde d}^2}.
  \end{aligned}
\end{equation}
which implies $\mathbb{E}[\hat e_p^2(\bm X)]=O(1/M)$.

By using Chebyshev's inequality, we get
\begin{equation}\label{equ:Chebyshev1}
  \begin{aligned}
  P\left \{ \left| \frac{\hat U_{\tilde d, p}(\bm X)}{V_{\tilde d}}-\frac{U_{\tilde d,p}(\bm X)}{V_{\tilde d}} \right| \ge \varepsilon \right \} \le \frac{ {\rm Var}\Big(\frac{\hat U_{\tilde d, p}(\bm X)}{V_{\tilde d}}\Big) }{\varepsilon^2},
  \end{aligned}
\end{equation}
where $\hat U_{\tilde d,p}(\bm X)= \frac{1}{M} \sum_{k=1}^M I_{\{\bm X_k\in \mathbb{U}_{\tilde d, \bm X} \} }$.
Let $\varepsilon = (\frac{1}{K})^{\eta/2}$ with some fixed $\eta \in (\frac{2}{3},1)$. In that case, we have
${\rm Var}(\frac{\hat U_{\tilde d, p}}{V_{\tilde d}}) / \varepsilon^2  = O(\frac{K^\eta}{M})$.
Thus, it is derived that
\begin{equation}\label{equ:hat e_p_1}
  \begin{aligned}
    & \hat e_p(\bm X)\\
    & =\hat e_p(\bm X)\big\{P\{|\hat e_p(\bm X)|<\varepsilon\}+P\{|\hat e_p(\bm X)|\ge \varepsilon\}\big\}\\
    & = \hat e_p(\bm X) I_{ \{|\hat e_p(\bm X)|<\varepsilon\} } P\{|\hat e_p(\bm X)|<\varepsilon\}
     +\hat e_p(\bm X) I_{ \{|\hat e_p(\bm X)|\ge \varepsilon\} } \} P\{|\hat e_p(\bm X)|<\varepsilon\}
    + O\Big(\frac{K^\eta}{M}\Big) \\
    & < \varepsilon + O\Big( \frac{K^\eta}{M} \Big)
    = O\Big( \big(\frac{1}{K}\big)^{\eta/2} \Big)+ O\Big( \frac{K^\eta}{M} \Big).
  \end{aligned}
\end{equation}

Furthermore, it is readily seen that
\begin{equation}\label{equ:E_hat e_p^m}
  \begin{aligned}
    \mathbb{E}[\hat e_p^m(\bm X)]
    & = \mathbb{E}\big [  I_{\{m=2\}}\hat e_p^2(\bm X) + I_{\{m\ge 3\}}\hat e_p^m(\bm X)
    \big ]\\
    & =I_{\{m=2\}}O\Big(\frac{1}{M} \Big)
    + I_{\{m\ge3\}}O\Big( \big(\frac{1}{K}\big)^{m\eta/2} \Big)\\
    & = I_{\{m=2\}}O\Big(\frac{1}{M} \Big)
    + I_{\{m\ge3\}}o\Big(\frac{1}{K} \Big).
  \end{aligned}
\end{equation}
Similarly, for $ \mathbb{E}[\hat e_q^n(\bm X)]$, we get the same result as $ \mathbb{E}[\hat e_p^m(\bm X)]$.

On the other hand, it is not difficult to see that the relationship between $U_{\tilde d,p}(\bm X)$ and $p(\bm X)$ is similar to $\hat U_{\tilde d,p}(\bm X)$ and $\hat p(\bm X)$ in Eq. (\ref{equ:expansion_hat_U}). Thus, it is easily seen that
\begin{equation}\label{equ:check_e_p_solution}
  \begin{aligned}
     \check e_p(\bm X)
    &=\frac{ U_{\tilde d,p}(\bm X)}{V_{\tilde d}}- p(\bm X)\\
    & = \sum_{i=1}^d \bar c_{i,\tilde d}(\bm X)\Big(\frac{K}{M}\Big)^{i/d}+ o\Big(\frac{K}{M}\Big).\\
  \end{aligned}
\end{equation}
Furthermore, we have
\begin{equation}\label{equ:E_check_e_p}
  \begin{aligned}
    & \mathbb{E}[\check e_p^m(\bm X)]
    = \left(\sum_{i=1}^d \check c_{e_p^m,i,\tilde d}\Big(\frac{K}{M}\Big)^{i/d}\right)^m +o\Big(\frac{K}{M}\Big).
  \end{aligned}
\end{equation}

As for $\bar e_p(\bm X)$, by applying the binomial theorem and Cauchy-Schwartz inequality, we have
\begin{equation}\label{equ:E_bar_e_p}
  \begin{aligned}
    \mathbb{E}[\bar e_p^m(\bm X)]
    & = \mathbb{E}\{[\hat e_p(\bm X)+ \check e_p(\bm X)]^m\}
    = \mathbb{E}[\sum_{j=0}^m \bar a_{p,j} \hat e_p^{j}(\bm X) \check e_p^{m-j}(\bm X) ]\\
    & \le \sum_{j=0}^m \bar a_{p,j} \sqrt{ \mathbb{E}[\hat e_p^{2j}(\bm X) ]
    \mathbb{E}[\check e_p^{2(m-j)}(\bm X) ] }\\
    & = \sqrt{ \mathbb{E}[\check e_q^{2m}(\bm X) ]}
    + \sum_{j=1}^m \bar a_{p,j} \sqrt{ \mathbb{E}[\hat e_p^{2j}(\bm X) ]
    \mathbb{E}[\check e_p^{2(m-j)}(\bm X) ] } \\
    & = \sum_{i=1}^d \bar c_{e_p^m,i,\tilde d}\Big(\frac{K}{M}\Big)^{i/d}
    +o\Big(\frac{K}{M}\Big)+o\Big( \frac{1}{ \sqrt{M} } \Big) + o\Big( \frac{1}{K} \Big),
  \end{aligned}
\end{equation}
where $\bar a_{p,j}$ is the binomial coefficient.

Similarly, we can derive $\mathbb{E}[\bar e_q^n(\bm X)]$, as well as, the proof of Lemma \ref{lem:E_bar_e^m_p_q} is completed.
\end{proof}
By expanding the $\frac{\bar q(\bm X)}{\bar p(\bm X)}$ around $q(\bm X)$ and $p(\bm X)$, we can make use of Lemma \ref{lem:E_bar_e^m_p_q} to derive $\check h(\bm X)=\frac{\bar q(\bm X)}{\bar p(\bm X)}-\frac{q(\bm X)}{p(\bm X)}$ and $\mathbb{E}[\check h^t(\bm X)]$, similar to Lemma
\ref{lem:E_h^t}.
Furthermore, using the same way as Eq. (\ref{equ:bias1_expension_solution}),
we can easily see that
\begin{equation}\label{equ:bias2_expension_solution}
  \begin{aligned}
    & \mathbb{E}\left[f\left(\frac{\bar q(\bm X)}{\bar p(\bm X)}\right)
    - f\left(\frac{ q(\bm X)}{ p(\bm X)}\right)\right]\\
    & =\mathbb{E}\left[ \sum_{i=1}^2\check c_{\check h_i,\tilde d}(\bm X)\check h^i(\bm X)
     + o\left(\check c_{\check h_3,\tilde d}(\bm X)\check h^2(\bm X)\right)\right]\\
    & = \sum_{i=1}^d c_{\check h,i,\tilde d}\Big(\frac{K}{M}\Big)^{i/d}
    + o\big(\frac{K}{M}\big)+ o\big(\frac{1}{K}\big).
  \end{aligned}
\end{equation}

Therefore, by combining Eq. (\ref{equ:bias1_expension_solution}) and Eq. (\ref{equ:bias2_expension_solution}), the proof of Theorem \ref{thm_bias} can be completed.
\section{Proof of Theorem \ref{thm_var}}\label{Appendix_thm_var}
For the function $f(x)=e^{\varpi x^{-1}}$ with $\varpi>0$ and $x>0$, we can do a Taylor series expansion of $f(\frac{\widetilde q_{\hat s_0}(\bm X)}{\widetilde p_{\hat s_0}(\bm X)})$ around
$\mathbb{E}[\frac{\widetilde q_{\hat s_0}(\bm X)}{\widetilde p_{\hat s_0}(\bm X)}]$ as following
\begin{equation}\label{equ:var1_expension}
 \begin{aligned}
    f\left(\frac{\widetilde q_{\hat s_0}(\bm X)}{\widetilde p_{\hat s_0}(\bm X)}\right)
    & =\sum_{i=0}^2 \frac{f^{(i)}\left( \mathbb{E}[\frac{\widetilde q_{\hat s_0}(\bm X)}{\widetilde p_{\hat s_0}(\bm X)}] \right)}{i!} \tilde \varrho^{i}(\bm X)
     + \frac{1}{6}f^{(3)}( \xi_{\bm X}) \tilde \varrho^{3}(\bm X),
  \end{aligned}
\end{equation}
where $ \xi_{\bm X}\in \left( \mathbb{E}[\frac{\widetilde q_{\hat s_0}(\bm X)}{\widetilde p_{\hat s_0}(\bm X)}],\frac{\widetilde q_{\hat s_0}(\bm X)}{\widetilde p_{\hat s_0}(\bm X)}
 \right)$ comes from the mean value theorem, and $\tilde \varrho(\bm X)$ denotes
$ \frac{\widetilde q_{\hat s_0}(\bm X)}{\widetilde p_{\hat s_0}(\bm X)}
-\mathbb{E}[\frac{\widetilde q_{\hat s_0}(\bm X)}{\widetilde p_{\hat s_0}(\bm X)}]$. Define the operator
$\mathcal{H}(\bm Z)=\bm Z-\mathbb{E}(\bm Z)$. Let
\begin{equation*}
 \begin{aligned}
   &{\bm a}_i= \mathcal{H}\bigg( f\Big(\mathbb{E}\big[\frac{\widetilde q_{\hat s_0}(\bm X_i)}{\widetilde p_{\hat s_0}(\bm X_i)}\big]\Big) \bigg),\qquad \\
   &{\bm b}_i= \mathcal{H}\bigg( f^{'}\Big(\mathbb{E}\big[\frac{\widetilde q_{\hat s_0}(\bm X_i)}{\widetilde p_{\hat s_0}(\bm X_i)}\big]\Big) \tilde \varrho(\bm X_i)\bigg), \qquad\\
 \end{aligned}
\end{equation*}
\begin{equation*}
 \begin{aligned}
   &{\bm c}_i= \mathcal{H}\bigg( \sum_{j=2}^{r-1}\frac{f^{(j)}\Big(\mathbb{E}\big[\frac{\widetilde q_{\hat s_0}(\bm X_i)}{\widetilde p_{\hat s_0}(\bm X_i)}\big]\Big)}{j!} \tilde \varrho^j(\bm X_i)\bigg),\\
  \end{aligned}
\end{equation*}
\begin{equation*}
 \begin{aligned}
   &{\bm d}_i= \mathcal{H}\bigg( \frac{f^{(r)}(\tilde \xi_{\bm X_i})}{r!}\tilde \varrho^r(\bm X_i) \bigg).\qquad \qquad
  \end{aligned}
\end{equation*}
Then the variance of $\hat F_{\tilde d}$ is
\begin{equation}\label{equ:Var_hat_F}
 \begin{aligned}
    {\rm Var}(\hat F_{\tilde d})
    & =\mathbb{E}[(\hat F_{\tilde d}-\mathbb{E}(\hat F_{\tilde d}))^2]\\
    & = \frac{1}{N}\mathbb{E}[({\bm a}_1+{\bm b}_1+{\bm c}_1+{\bm d}_1)^2]
    + \frac{N-1}{N}\mathbb{E}[({\bm a}_1+{\bm b}_1+{\bm c}_1+{\bm d}_1)({\bm a}_2+{\bm b}_2+{\bm c}_2+{\bm d}_2)],
  \end{aligned}
\end{equation}
which can be bounded in the following.
\begin{lem}\label{lem:E_e_p^m_e_q^n}
For a $d$-dimension variable $\bm X$, a realization of p.m.f $p$ independent of the window kernel estimators $\widetilde q_{\hat s_0}$ and $\widetilde p_{\hat s_0}$,
it can be given that
\begin{subequations}\label{equ:E_e_p_q}
  \begin{align}
    &
    \begin{aligned}
    & \mathbb{E}[ e_p^m(\bm X)]
    = \mathbb{E}\big\{\big[ \tilde p_{\hat s_0}(\bm X)-\mathbb{E}[\tilde p_{\hat s_0}(\bm X)] \big]^m\big\}
    = o\Big( \frac{1}{K} \Big)
    \end{aligned}\\
    &
    \begin{aligned}
    & \mathbb{E}[e_q^n(\bm X)]
    = \mathbb{E}\big\{\big[\tilde q_{\hat s_0}(\bm X)-\mathbb{E}[\tilde q_{\hat s_0}(\bm X)]\big]^n\big\}
    = o\Big( \frac{1}{K} \Big)
    \end{aligned}
  \end{align}
\end{subequations}
where $K=\varDelta \sqrt{M}$ with $\varDelta=[\frac{(2\tilde d_l+1)}{L}]^d \sqrt{M}$, $ e_{p}(\bm X)$ and $e_{q}(\bm X)$ denote \{$\tilde p_{\hat s_0}(\bm X)-\mathbb{E}[\tilde p_{\hat s_0}(\bm X)] $\} and \{$\tilde q_{\hat s_0}(\bm X)-\mathbb{E}[\tilde q_{\hat s_0}(\bm X)]$\}, respectively.
\end{lem}
\begin{proof}[{~~~Proof:}]
For the $e_p(\bm X)=\tilde p_{\hat s_0}(\bm X)-\mathbb{E}[\tilde p_{\hat s_0}(\bm X)]$, it is readily seen that
\begin{equation}\label{equ:E_e_p^1}
  \begin{aligned}
  & \mathbb{E}[e_p(\bm X)]=\mathbb{E}\big[\tilde p_{\hat s_0}(\bm X)-\mathbb{E}[\tilde p_{\hat s_0}(\bm X)] \big]=0.\\
  \end{aligned}
\end{equation}
In addition, 
by use of the definition \ref{defn:3}, we have
\begin{equation}\label{equ:E_e_p^2}
  \begin{aligned}
    & {\rm Var}\Big(\tilde p_{\hat s_0}(\bm X)\Big)\\
    & = \mathbb{E}\big[\tilde p_{\hat s_0}^2(\bm X)\big] - \mathbb{E}^2\big[\tilde p_{\hat s_0}(\bm X)\big] \\
    & =\frac{1}{M} \sum_{\bm x_j \in \mathbb{U}_{\tilde d,\bm X} }W^2(\hat s_0, \bm X, \bm x_j) p(\bm x_j)
    -\frac{1}{M}\left[ \sum_{\bm x_j \in \mathbb{U}_{\tilde d,\bm X} }W(\hat s_0, \bm X, \bm x_j)p(\bm x_j)\right]^2,
  \end{aligned}
\end{equation}
where $\mathbb{U}_{\tilde d,\bm X}=\{\bm x_j: | x_j^u - X^u|\le \tilde d, \bm x_j\in \mathbb{U} \}$. Consequently, it is evident to see that $\mathbb{E}[e_p^2(\bm X)]=O(1/M)$.

Similar to Eq. (\ref{equ:hat e_p_1}) with parameter $\varepsilon = (\frac{1}{K})^{\eta/2}$ and $\eta \in (\frac{2}{3},1)$, we have
\begin{equation}\label{equ:e_p_1}
  \begin{aligned}
    & e_p(\bm X)\\
    & = e_p(\bm X)\big\{P\{| e_p(\bm X)|<\varepsilon\}+P\{| e_p(\bm X)|\ge \varepsilon\}\big\}\\
    & = e_p(\bm X) I_{ \{|e_p(\bm X)|< \varepsilon\} } \} P\{|e_p(\bm X)|<\varepsilon\}
    + O\Big(\frac{K^\eta}{M}\Big) \\
    & < \varepsilon + O\Big( \frac{K^\eta}{M} \Big)
    = O\Big( \big(\frac{1}{K}\big)^{\eta/2} \Big)+ O\Big( \frac{K^\eta}{M} \Big).
  \end{aligned}
\end{equation}
Consequently, we get
\begin{equation}\label{equ:E_e_p^m}
  \begin{aligned}
    & \mathbb{E}[e_p^m(\bm X)]\\
    & = \mathbb{E}\big [  I_{\{m=2\}}e_p^2(\bm X) + I_{\{m\ge 3\}}e_p^m(\bm X)
    \big ]\\
    & =I_{\{m=2\}}O\Big(\frac{1}{M} \Big)
    + I_{\{m\ge3\}}O\Big( \big(\frac{1}{K}\big)^{m\eta/2} \Big)\\
    & =I_{\{m=2\}}O\Big(\frac{1}{M} \Big) + I_{\{m\ge3\}}o\Big(\frac{1}{K} \Big).
  \end{aligned}
\end{equation}
Similarly, it can be proved that $ \mathbb{E}[e_q^n(\bm X)]$ has the same result as $ \mathbb{E}[ e_p^m(\bm X)]$.
\end{proof}
By using a Taylor series expansion of $1/\widetilde p_{\hat s_0}(\bm X)$ around $\mathbb{E}[\widetilde p_{\hat s_0}(\bm X)]$, we have
\begin{equation}\label{equ:expension_E_p_s0}
  \begin{aligned}
    \mathbb{E}[\frac{1}{\widetilde p_{\hat s_0}(\bm X)}]
    & = \mathbb{E}\Big[\frac{1}{\mathbb{E}[\widetilde p_{\hat s_0}(\bm X)]}
    -\frac{e_p(\bm X)}{\mathbb{E}^2[\widetilde p_{\hat s_0}(\bm X)]}
    + \frac{e_p^2(\bm X)}{\tilde \xi_{p, \bm X}^3 }\Big]\\
    &= \frac{1}{\mathbb{E}[\widetilde p_{\hat s_0}(\bm X)]}
    -\frac{\mathbb{E}[e_p(\bm X)]}{\mathbb{E}^2[\widetilde p_{\hat s_0}(\bm X)]}
    + o(\mathbb{E}[e_p^2(\bm X)]),
  \end{aligned}
\end{equation}
where $e_{p}(\bm X)$ denotes \{$\widetilde p_{\hat s_0}(\bm X)-\mathbb{E}[\widetilde p_{s_0}(\bm X)]$\}, and $\tilde \xi_{p, \bm X}\in \left(\mathbb{E}[\widetilde p_{s_0}(\bm X)], \widetilde p_{\hat s_0}(\bm X)\right)$.
Since the $p$ and $q$ are independent, it is easily to see that
\begin{equation}\label{equ:expension_E_q_s0/p_s0}
  \begin{aligned}
     \mathbb{E}[\frac{\widetilde q_{\hat s_0}(\bm X)}{\widetilde p_{\hat s_0}(\bm X)}]
    & = \mathbb{E}[\widetilde q_{\hat s_0}(\bm X)]\mathbb{E}[\frac{1}{\widetilde p_{\hat s_0}(\bm X)}]\\    & = \frac{\mathbb{E}[\widetilde q_{\hat s_0}(\bm X)]}{\mathbb{E}[\widetilde p_{\hat s_0}(\bm X)]}
    + o(\mathbb{E}[e_p^2(\bm X)]).
  \end{aligned}
\end{equation}

What is more, by expanding the $\frac{\widetilde q_{s_0}(\bm X)}{\widetilde p_{s_0}(\bm X)}$ around
$\mathbb{E}[\widetilde q_{s_0}(\bm X)]$ and $\mathbb{E}[\widetilde p_{s_0}(\bm X)]$, we have
\begin{equation}\label{equ:E_tilde_varrho}
  \begin{aligned}
    & \tilde \varrho(\bm X)\\
    & = \frac{e_q(\bm X)}{\mathbb{E}[\widetilde p_{s_0}(\bm X)]}
    - \frac{\mathbb{E}[\widetilde q_{s_0}(\bm X)]}{\mathbb{E}^2[\widetilde p_{s_0}(\bm X)]}e_p(\bm X)
     - \frac{ e_q(\bm X) e_p(\bm X)}{\mathbb{E}^2[\widetilde p_{s_0}(\bm X)]}
     + \frac{\mathbb{E}[\widetilde q_{s_0}(\bm X)]}{\mathbb{E}^3[\widetilde p_{s_0}(\bm X)]} e_p^2(\bm X)\\
    & + \frac{\mathbb{E}[\widetilde q_{s_0}(\bm X)]}{\mathbb{E}^4[\widetilde p_{s_0}(\bm X)]} e_q(\bm X)e_p^2(\bm X)+ o(\mathbb{E}[e_p^2(\bm X)])
     + o(e_p^2(\bm X)+ e_q(\bm X)e_p^2(\bm X)),
  \end{aligned}
\end{equation}

In addition, by applying the Cauchy-Schwartz inequality and Lemma \ref{lem:E_e_p^m_e_q^n}, it is readily seen that
\begin{equation}\label{equ:E_varrho_e_p^mq^n}
  \begin{aligned}
    & |\mathbb{E}[e_p^m(\bm X)e_q^n(\bm X)]|\\
    & \le I_{\{m\ge 1, n\ge 1\}}\sqrt{\mathbb{E}[e_p^{2m}(\bm X)]\mathbb{E}[e_q^{2n}(\bm X)]}
     + \big \{ I_{\{m> 1, n=0\}} + I_{\{m= 0, n>1\}} \big\}|\mathbb{E}[e_p^m(\bm X)e_q^n(\bm X)]|\\
    & + \big \{ I_{\{m= 1, n=0\}} + I_{\{m= 0, n=1\}} \big\}|\mathbb{E}[e_p^m(\bm X)e_q^n(\bm X)]|\\
    & =o\Big(\frac{1}{\sqrt{M}}\Big)+o\Big(\frac{1}{K}\Big),
  \end{aligned}
\end{equation}
where $m>0$ and $n>0$ are integers.

Therefore, by applying the multinomial theorem to Eq. (\ref{equ:E_tilde_varrho}), it is apparent that
\begin{equation}\label{equ:expension_varrho_solution}
\begin{aligned}
    \mathbb{E}[\tilde \varrho^{t}(\bm X)]
    & =\mathbb{E} \bigg \{ \Big\{\frac{\widetilde q_{\hat s_0}(\bm X)}{\widetilde p_{\hat s_0}(\bm X)}
    -\mathbb{E}\big[\frac{\widetilde q_{\hat s_0}(\bm X)}{\widetilde p_{\hat s_0}(\bm X)}\big] \Big\}^t \bigg\}\\
    & = o\Big(\frac{1}{\sqrt{M}}\Big)+ o\Big(\frac{1}{K}\Big),
\end{aligned}
\end{equation}
According to Eq. (\ref{equ:expension_varrho_solution}) and Eq. (\ref{equ:Var_hat_F}), we have
\begin{equation}\label{equ:Var_a1_d1}
 \begin{aligned}
    \mathbb{E}[({\bm a}_1+{\bm b}_1+{\bm c}_1+{\bm d}_1)^2]
    & = \mathbb{E}[{\bm a}_1^2]+o(1)
    = c_N \Big(\mathbb{E}\big[\frac{\widetilde q_{\hat s_0}(\bm X_i)}{\widetilde p_{\hat s_0}(\bm X_i)}\big] \Big) +o(1).
  \end{aligned}
\end{equation}
Moreover, due to the independent $\bm X_1$ and $\bm X_2$, it can be readily seen that $\mathbb{E}[({\bm a}_1+{\bm b}_1+{\bm c}_1+{\bm d}_1)({\bm a}_2+{\bm b}_2+{\bm c}_2+{\bm d}_2)]=0$.
Therefore, we have
\begin{equation}\label{equ:Var_hat_F_solution}
 \begin{aligned}
    & {\rm Var}(\hat F_{\tilde d})
    = \frac{1}{N}\Big[c_N \Big(\mathbb{E}\big[\frac{\widetilde q_{\hat s_0}(\bm X_i)}{\widetilde p_{\hat s_0}(\bm X_i)}\big] \Big) +o(1)\Big]
    = O\Big(\frac{1}{N}\Big),
  \end{aligned}
\end{equation}
which indicates that Theorem \ref{thm_var} is proved.

\section{The {\rm MSE} convergence of M-I divergence estimation}\label{Appendix_MSE_M-I}
According to Eq. (\ref{equ:hat_F_d}) and Taylor series expansion, we have
\begin{equation}\label{equ:hat_F-F}
\begin{aligned}
  \hat F_{\tilde d}-F_\varpi(P\parallel Q)
  & = \frac{1}{N} \sum_{i=1}^N [e^{ \varpi \frac{\widetilde p_{\hat s_0}(\bm X_i)}
  {\widetilde q_{\hat s_0}(\bm X_i)} }- e^{ \varpi \frac{ p(\bm X_i)}
  {q(\bm X_i)} } ]
  +\Big\{\frac{1}{N} \sum_{i=1}^N e^{ \varpi \frac{ p(\bm X_i)}{q(\bm X_i)} }
  - \sum_{\bm x_j\in \mathbb{U} } p(\bm x_j)e^{ \varpi \frac{ p(\bm x_j)}{q(\bm x_j)} } \Big\}\\
  &= \frac{1}{N} \sum_{i=1}^N \bigg\{\varpi e^{ \varpi \frac{ p(\bm X_i)}{q(\bm X_i)} }
  \Big[\frac{\widetilde p_{\hat s_0}(\bm X_i)}{\widetilde q_{\hat s_0}(\bm X_i)}
  -\frac{p(\bm X_i)}{q(\bm X_i)}\Big]
   + O\Big(\big(\frac{\widetilde p_{\hat s_0}(\bm X_i)}{\widetilde q_{\hat s_0}(\bm X_i)}
  -\frac{p(\bm X_i)}{q(\bm X_i)}\big)^2\Big)\bigg\}\\
  &+\Big\{\frac{1}{N} \sum_{i=1}^N e^{ \varpi \frac{ p(\bm X_i)}{q(\bm X_i)} }
  - \sum_{\bm x_j\in \mathbb{U} } p(\bm x_j)e^{ \varpi \frac{ p(\bm x_j)}{q(\bm x_j)} } \Big\}.
\end{aligned}
\end{equation}

Considering Remark \ref{rem:tilde_p}
and the law of large numbers, 
it is easy to see that $\hat F_{\tilde d}-F_\varpi(P\parallel Q)\to 0$. Therefore, we have
\begin{equation}\label{equ:hat_F_d-F}
\begin{aligned}
     &\hat F_\lambda-F_\varpi(P\parallel Q)=\sum_{l\in \bar l}\lambda_l [\hat F_{\tilde d_l}-F_\varpi(P\parallel Q)] \to 0,
\end{aligned}
\end{equation}
which equally means $\hat F_\lambda/F_\varpi(P\parallel Q)\to 1$.

Additionally, in virtue of corollary \ref{cor_ensemble} and the definition of ${\rm MSE}$, it is readily seen that
\begin{equation}\label{equ:ensemble_F*}
\begin{aligned}
   \mathbb{E}\left\{\Big[ \frac{\hat F_{\lambda^*} }{F_\varpi(P \parallel Q)}-1 \Big]^2\right\}=
    \frac{{\rm MSE}(\hat F_{\lambda^*})}{F_\varpi^2(P \parallel Q)}=O(\varGamma^{-1}),
\end{aligned}
\end{equation}
where $\hat F_{\lambda^*}$ is the optimal ensemble estimation of $F_\varpi(P \parallel Q)$, and $F_\varpi(P \parallel Q) < \infty$.

Therefore, by using the Taylor series expansion of $\log(x)$ around $x=1$, we have
\begin{equation}\label{equ:ensemble_MSE_log}
\begin{aligned}
   & {\rm MSE}(\hat D_\varpi(P\parallel Q))\\
   &=\mathbb{E}\left\{\Big[ \log \frac{\hat F_{\lambda^*} }{F_\varpi(P \parallel Q)} \Big]^2\right\}\\
   &= \mathbb{E}\Bigg\{\bigg[ \Big(\frac{\hat F_{\lambda^*} }{F_\varpi(P \parallel Q)}-1\Big)
   -\frac{1}{2}\Big(\frac{\hat F_{\lambda^*} }{F_\varpi(P \parallel Q)}-1\Big)^2
   + o\bigg(\Big(\frac{\hat F_{\lambda^*} }{F_\varpi(P \parallel Q)}-1\Big)^2\bigg)
   \bigg]^2\Bigg\}\\
   &=\mathbb{E}\bigg\{\Big[ \frac{\hat F_{\lambda^*} }{F_\varpi(P \parallel Q)}-1 \Big]^2
   + o\bigg( \Big( \frac{\hat F_{\lambda^*} }{F_\varpi(P \parallel Q)}-1 \Big)^2 \bigg)\bigg\}
   =O(\varGamma^{-1}),
\end{aligned}
\end{equation}
which proves Eq. (\ref{equ:ensemble_MSE_hat_D}).

%
%
%
%

\section*{Acknowledgment}
The authors would like to thank a lot for the support of the China Major State Basic Research Development Program (973 Program) No.2012CB316100(2), National Natural Science Foundation of China (NSFC) No. 61771283 and China Scholarship Council. Meanwhile, the work received the extensive comments on the final draft from other members of Wistlab of Tsinghua University.
\ifCLASSOPTIONcaptionsoff
  \newpage
\fi

\end{document}